\documentclass[12pt]{article}

\usepackage[T1]{fontenc}
\usepackage[utf8]{inputenc}
\usepackage[english]{babel}

\usepackage{amsmath}
\usepackage{amssymb}
\usepackage{amsthm}
\usepackage[hidelinks]{hyperref}

\usepackage{enumitem}
\usepackage{booktabs}
\usepackage{graphicx}
\usepackage{mathtools}
\mathtoolsset{showonlyrefs,showmanualtags}
\usepackage{natbib}

\usepackage{rotating}

\usepackage{bbm}

\usepackage[format=plain,font=small,labelfont=bf,textfont=up]{caption}

\newtheorem{theorem}{Theorem}

\newtheorem{lemma}[theorem]{Lemma}

\newtheorem{innercustomthm}{Theorem}
\newenvironment{customtheorem}[1]
{\renewcommand\theinnercustomthm{#1}\innercustomthm}
{\endinnercustomthm}

\newtheorem{innercustomlma}{Lemma}
\newenvironment{customlemma}[1]
{\renewcommand\theinnercustomlma{#1}\innercustomlma}
{\endinnercustomlma}

\theoremstyle{definition}
\newtheorem{definition}[theorem]{Definition}
\newtheorem{remark}[theorem]{Remark}

\DeclareMathOperator{\N}{N}
\DeclareMathOperator{\NN}{NN}
\DeclareMathOperator{\dm}{d}
\DeclareMathOperator*{\logistic}{logistic}

\DeclareMathOperator{\mean}{mean}

\title{\textbf{Generalized full matching and extrapolation of the results from a large-scale voter mobilization experiment}\thanks{We thank Peter Aronow, Justin Grimmer, Ben Hansen, Kosuke Imai, Walter Mebane, Sam Pimentel, Liz Stuart, Cyrus Samii, Yotam Shem-Tov and Jos\'{e} Zubizarreta for helpful suggestions and discussions. We thank Alan Gerber, Don Green and Christopher Larimer for sharing the data from their experiment. The work was partially supported by Office of Naval Research Grants N00014-15-1-2367 and N00014-17-1-2176.}}

\author{%
	Fredrik Sävje\thanks{Department of Political Science and Department of Statistics \& Data Science, Yale University.}
	\and
	Michael J. Higgins\thanks{Department of Statistics, Kansas State University.}
	\and
	Jasjeet S. Sekhon\thanks{Travers Department of Political Science and Department of Statistics, UC Berkeley.}
}

\date{\today}

\begin{document}

\def\spacingset#1{\renewcommand{\baselinestretch}%
{#1}\small\normalsize} \spacingset{1}

	\maketitle

	\bigskip
	\begin{abstract}
	\noindent Matching is an important tool in causal inference.
The method provides a conceptually straightforward way to make groups of units comparable on observed characteristics.
The use of the method is, however, limited to situations where the study design is fairly simple and the sample is moderately sized.
We illustrate the issue by revisiting a large-scale voter mobilization experiment that took place in Michigan for the 2006 election.
We ask what the causal effects would have been if the treatments in the experiment were scaled up to the full population.
Matching could help us answer this question, but no existing matching method can accommodate the six treatment arms and the 6,762,701 observations involved in the study.
To offer a solution this and similar empirical problems, we introduce a generalization of the full matching method and an associated algorithm.
The method can be used with any number of treatment conditions, and it is shown to produce near-optimal matchings.
The worst case maximum within-group dissimilarity is no worse than four times the optimal solution, and simulation results indicate that its performance is considerably closer to the optimal solution on average.
Despite its performance, the algorithm is fast and uses little memory.
It terminates, on average, in linearithmic time using linear space.
This enables investigators to construct well-performing matchings within minutes even in complex studies with samples of several million units.

	\end{abstract}

	\vfill

	\spacingset{1.45}

	\section{Introduction}

\subsection{A voter mobilization experiment}

Political scientists and economists have long been puzzled by the fact that voters vote.
The probability that any particular voter is pivotal in an election is negligible, so the benefit of voting appears to be small, but the cost is not.
In a landmark study, Gerber,~Green~\&~Larimer~(\citeyear{Gerber2008Social}) investigate whether social norms can provide an explanation.
The authors randomly assigned 344,084 registered voters to one of five treatment conditions in the 2006 primary election in Michigan.
The condition of main interest was the receipt of a postcard listing the voting history of the recipients and their neighbors.
The recipients were also informed that updated postcards would be sent out after the election.
The purpose was to use social pressure to motivate recipients to vote; if abstaining, their neighbors would know that they did not fulfill their civic duty, and they may suffer social costs or stigma.
The turnout among recipients of the postcard was $37.8\%$.
This is to be compared with a turnout of $29.7\%$ in the control group, who did not receive a postcard.
The estimated causal effect is therefore $8.1$ percentage points, indicating that social pressure was a motivation for these voters.

The Michigan experiment is impressive both in scale and design, but it has one important shortcoming.
The sample used in the experiment is not representative of the overall population.
The turnout among all registered voters in Michigan was $17.7\%$ in the 2006 primary election.
We would expect this figure to be close to the $29.7\%$ turnout in the control group if the sample was representative of the population.
The objective of the experiment was to establish whether social pressure can be a determinant of voting, so the authors constructed a sample with individuals deemed to be receptive to the postcard intervention.
The practice is methodologically sound because it maximize power with respect to the question at hand, but a consequence is that an investigation of voting behavior more generally becomes less tractable.
A careful approach, adjusting for the systematic differences between the sample and population, is needed to extrapolate the result from the experiment.
This is the task we undertake in this paper.

The exercise of generalizing the results from an experiment has attracted much recent interest \citep[see, e.g.,][]{Stuart2011,Hartman2015,Kern2016,Andrews2017,Buchanan2018,Dehejia2019}.
The typical approach is based on the assumption that all factors used to construct the experimental sample are observed.
If this indeed is the case, methods traditionally used to account for confounded treatment assignment in observational studies can be used for the extrapolation.
Units not included in the experimental study can be seen as assigned to an alternative treatment condition, and approaches of covariate adjustments between treatment groups can be applied.
The results from the extrapolation are less reliable than the results from the experiment itself because we rarely know what factors determined the sample.
The Michigan experiment is, however, an exception in this regard.
The construction of the experimental sample was based on the voter file, containing a record for every registered voter in Michigan, and we have access to this data set.
In other words, the selection-on-observables assumption, as it often is called, is known to be satisfied.

Conceptually, the task ahead is simple.
We need to make the experimental sample comparable to the overall population with respect to observed characteristics in the voter file.
Practically, the task is far from trivial.
The first challenge is the type of information contained in the voter file.
The file only contains information necessary to organize fair elections, so it is fairly sparse, but it does contain the addresses of the registered voters.
The authors of the original study worked with a political consultant to construct the sample using proprietary indices of partisanship and voting behavior.
We know that these indices were constructed based on the information in the voter file, including geographical coordinates derived from the addresses, but we do not know how the information was used.
While we in principle can adjust for the systematic selection because all information used in the sample selection process is observed, the relevant functional form of the coordinates in the indices is unknown and likely complex.
In particular, we cannot rule out that auxiliary geographical information have been merged to the voter file and subsequently used in the construction of the sample.

The need to accommodate geographical information limits how the covariate adjustment can be done.
A common approach is to assign weights to the observations so to make selected covariate moments similar between the treatment groups after reweighing \citep[see, e.g.,][]{Graham2012,Hainmueller2012,Diamond2013,Imai2014,Zubizarreta2015,Athey2018}.
The approach generally works well, but it might not be a good choice in the current application because there is no obvious way to use the geographical information.
It is infeasible to track down all information that could have been merged to the voter file, so the coordinates must be used directly.
However, even a large number of moments and cross-moments are unlikely to capture all aspects of the coordinates relevant for the sample selection process.

We will use a matching method to address the issue \citep{Cochran1973}.
The approach constructs groups of units that are as homogeneous as possible with respect to the observed confounders while still having all treatment conditions represented in each matched group.
The homogeneity of units is judged using an arbitrary distance metric, which easily can accommodate the geographical coordinates in a flexible, non-parametric way.
The adjustment is achieved by assigning weights the units proportionally to the relative prevalence of the treatment conditions within the matched groups.
The reweighed treatment groups will be comparable if the matched groups are sufficiently homogeneous.
Unlike the moment-based approaches referenced above, the treatment groups will be similar with respective to the whole joint covariate distribution.
This comes at the cost of less balance on specific moments compared to when the moments are specifically targeted.

The choice of matching-based adjustments for the analysis leads us to the second challenge.
The experiment was large and complex with 344,084 participants and five treatment conditions.
When adding the overall population from the complete voter file, the data set consists of 6,762,701 observations divided between six effective treatment conditions.
The typical application for matching methods is a study with two treatment conditions and a few thousand observations.
This motivates the development of a new matching method.

\subsection{Methodological contribution}

The conceptual simplicity of matching is deceiving.
The approach involves an intricate integer programming problem.
It is generally intractable to derive optimal solutions to these problems unless the sample is small and the design is simple \citep{Zubizarreta2012}.
This paper details a matching method with near-optimal performance accommodating complex studies with large samples.

The introduced method is a generalization of \emph{full matching}.
Full matching is a flexible method that is optimal for a large set of common cases \citep{Rosenbaum1991,Hansen2004}.
In particular, among all matching methods that do not leave units unassigned, full matching has the least within-group heterogeneity.
The method is, however, restricted to studies with two treatment conditions where the investigator requires no more than one unit of each treatment condition in the matched groups.
Subsequently, full matching cannot be used with more complex designs, and investigators are forced to resort to suboptimal solutions in these situations.
Two such designs are studies with more than two treatment conditions and when the matched groups are required to contain more than two units of each treatment condition.
The method we introduce, \emph{generalized full matching}, allows for designs with multiple treatments and complex compositional restrictions.

Existing matching algorithms cannot be used to derive generalized full matching.
The most widely used algorithm for full matching \citep{Hansen2006} derives optimal solutions, and it is thus an excellent choice when it can be used.
Its focus is, however, traditional designs with two treatment conditions.
The derivation of optimal solutions is also computationally demanding, and the algorithm cannot be used with large samples.
The core contribution of this paper is the development of an algorithm to derive generalized full matchings.

The algorithm derives near-optimal generalized full matchings in a wide range of settings, and it does so quickly even in large samples.
A matching produced by the algorithm is guaranteed to be within a constant factor of the optimal solution, ensuring that the matching is reasonable.
Simulations show that the derived matchings on average perform close to on par with the optimal full matching algorithm in cases where the optimal algorithm can be used.
The algorithm scales well and terminates in linearithmic time on average.
The simulation study shows that the algorithm is several orders of magnitude faster than existing solutions.

The subsequent sections describe the generalized full matching method and the associated algorithm.
The paper concludes by returning to the Michigan voter mobilization experiment to investigate what the effect would have been if the treatments were scaled up to the complete population.

\section{Generalized full matching} \label{sec:genfulmatch}

\subsection{Background}

Matching methods make treatment groups comparable by down-weighting, implicitly or explicitly, units that have a treatment assignment that is overrepresented given their characteristics.
That is, units assigned to a treatment condition that is uncommon locally in the covariate space are given a larger weight than units with a common condition.
As we saw in the Michigan experiment, people with a higher baseline propensity to vote were overrepresented in the experiment, so they must be given a smaller weight to make the experimental sample representative of the population.

We might be able to perfectly equalize the covariate distribution between the treatment groups when the confounders are few and coarse.
That is, we can construct an \emph{exact matching}.
This is achieved by stratifying the sample based on the confounders so that all units within a matched group are identical.
Exact matchings are rarely feasible because balance is often sought on continuous and other high-dimensional variables.
The units are in these cases partitioned into groups that are as homogeneous as possible, but not necessarily identical, producing an approximate matching.

The construction of the matched groups involves several considerations.
The first and immediate consideration is the objective of the matching, namely to make the matched groups as homogeneous as possible.
Homogeneity is typically assessed through pairwise distances between units based on some distance metric deemed relevant for the application at hand.
Common choices are the absolute difference between propensity scores \citep{Rosenbaum1983} and Euclidean and Mahalanobis distances in the covariate space \citep{Cochran1973}.

An equally important consideration is the composition of the matched groups.
The archetypical restriction on the composition is \emph{nearest neighbor matching} (also called 1:1-matching).
Each matched group is here required to contain exactly one treated unit and exactly one control unit.
The matching can be done with replacement, where the same units can be matched to several other units, or without replacement.
In both cases, units without matches are discarded.
Nearest neighbor matching often yields homogeneous groups, but the approach comes with the obvious disadvantage that large parts of the sample may be ignored.

\citet{Rosenbaum1991} introduced \emph{full matching} to address the issue.
The method imposes two compositional constraints.
First, all units must be assigned to a matched group; no units can be discarded.
Second, all groups must contain at least one unit of each treatment condition.
\citeauthor{Rosenbaum1991} studies this type of matching in settings with two treatment conditions, and he shows that all matched groups in the optimal matching under the two constraints will contain exactly one unit of at least one treatment conditions.
The insight allows him to construct an algorithm to construct optimal matchings for a wide range of distance metrics.
The method allows investigators to construct matched groups of high quality without discarding units.
\citet{Hansen2004} provided important developments, which we discuss in more detail in the concluding remarks.

The conventional formulation of full matching requires a particular design.
It can only be used in studies with two treatment conditions when the investigator accepts matched groups with only two units.
While most observational studies conform to this design, many do not.
The conventional formulation is unsatisfactory in more complex settings.
Examples include when there are several treatment conditions or when larger matched groups are needed for heterogeneous treatment effect analysis or variance estimation.
Currently, such studies must use cruder matching methods which might introduce bias or increase variance.
The following section introduces a generalization of conventional full matching that can be used in these more complicated settings.

\subsection{A generalization of full matching}

Consider a sample consisting of $n$ units indexed by $\mathbf{U} = \{1, 2, \cdots, n\}$.
The units have been assign to one of $k$ treatment conditions indexed by $\{1, 2, \cdots, k\}$ through an unknown or partially unknown process.
Let $W_i$ denote the condition that unit $i$ is assigned to. We construct a set for each condition, $\mathbf{w}_1, \mathbf{w}_2, \cdots, \mathbf{w}_k$, collecting the units assigned to the corresponding treatment: $\mathbf{w}_x = \{i\in\mathbf{U}: W_i = x\}$.

A matched group, $\mathbf{m}$, is a non-empty set of unit indices.
A matching, $\mathbf{M}$, is a set of disjoint matched groups: $\mathbf{M} = \{\mathbf{m}_1, \mathbf{m}_2, \cdots\}$.
A matching problem is defined by a set of constraints and an objective function.
The constraints describe a collection of admissible matchings, $\mathcal{M}$, and the objective function maps from the admissible matchings to a real-valued measure of match quality: $L: \mathcal{M} \rightarrow \mathbb{R}$.

\begin{definition} \label{def:ogfm}
	An \emph{optimal matching} $\mathbf{M}^*$ is an admissible matching that minimizes the matching objective:
	\begin{equation*}
		L(\mathbf{M}^*)= \min_{\mathbf{M}\in\mathcal{M}} L(\mathbf{M}).
	\end{equation*}
\end{definition}

Generalized full matching imposes the constraint that each unit is assigned to exactly one group.
The investigator can also impose constraints on the composition of the matched groups.
In particular, for each treatment condition $x$, one can require that each matched group contains at least $c_x$ units assigned to the corresponding condition.
One can also require that each group contains at least $t$ units in total (irrespectively of treatment assignment).

\begin{definition} \label{def:agfm}
	An \emph{admissible generalized full matching} for constraints $\mathcal{C}=(c_1, \cdots, c_k, t)$ is a matching $\mathbf{M}$ that satisfies:
	\begin{enumerate}
		\item (Spanning) $\bigcup_{\mathbf{m}\in\mathbf{M}}\mathbf{m} = \mathbf{U}$,
		\item (Disjoint) $\forall \mathbf{m}, \mathbf{m}'\in\mathbf{M}, \mathbf{m} \neq \mathbf{m}' \Rightarrow \mathbf{m} \cap \mathbf{m}' = \emptyset$,
		\item (Treatment constraints) $\forall \mathbf{m}\in\mathbf{M}, \forall x \in \{1, \cdots, k\}, |\mathbf{m} \cap \mathbf{w}_x| \geq c_x$,
		\item (Overall constraint) $\forall \mathbf{m}\in\mathbf{M}, |\mathbf{m}| \geq t$.
	\end{enumerate}
	Let $\mathcal{M}_\mathcal{C}$ collect all admissible generalized full matchings for constraints $\mathcal{C}$.
\end{definition}

As an example, consider a study with three treatment conditions.
The constraint $\mathcal{C}=(2, 2, 4, 10)$ would restrict the admissible matchings to those where each matched group contains two units of the first and second treatment conditions, four units of the third condition and ten units in total.

When we restrict the matching constraints so to require only one unit of each treatment condition in the matched groups, we recover the conventional full matching definition for studies with an arbitrary number of treatment conditions.

\begin{definition} \label{def:otfm}
A \emph{conventional full matching} in a study with $k$ treatment conditions is a generalized full matching with the matching constraints $\mathcal{C}=(1, 1, \cdots, 1, k)$.
\end{definition}

Our definition of full matching differs slightly from the original definition in \citet{Rosenbaum1991}.
For studies with two treatment conditions, the conventional definition requires, in addition the conditions in Definition \ref{def:otfm}, that each matched group contains exactly one treated unit or exactly one control unit:
\begin{equation*}
\forall \mathbf{m}\in\mathbf{M}, |\mathbf{m} \cap \mathbf{w}_1| = 1 \vee |\mathbf{m} \cap \mathbf{w}_2| = 1.
\end{equation*}
However, as detailed in Proposition 1 in \citet{Rosenbaum1991}, the optimal generalized full matching with constraints $\mathcal{C}=(1, 1, 2)$ is by necessity a full matching according to the original definition.
As a result, we can disregard the additional conditions imposed by \citet{Rosenbaum1991} and equivalently define full matchings as the optimal solution to the broader class of matching problems given by Definition \ref{def:agfm}.

\subsection{Near-optimal matchings}

As the number of units in a matching problem grows large, it may become intractable to derive optimal solutions.
In fact, generalized full matching is an \textbf{NP}-hard problem \citep{Higgins2016}.
However, as we will see in the next section, the matching problem becomes tractable if we allow for approximately optimal generalized full matchings.
That is, matchings that are guaranteed to be within some factor of the optimal solution.

\begin{definition} \label{def:aogfm}
	An \emph{$\alpha$-approximate matching} $\mathbf{M}^\dagger$ is an admissible matching that is within a factor of $\alpha$ of an optimal matching: $L(\mathbf{M}^\dagger) \leq \alpha L(\mathbf{M}^*)$.
\end{definition}

\section{An algorithm for generalized full matchings}

We now turn to the description of the algorithm used to construct near-optimal generalized full matchings.
The algorithm is an extension of the blocking algorithm introduced by \citet{Higgins2016}.
Blocking is an experimental design where similar units are grouped together into blocks and treatment is assigned within the blocks.
Matching and blocking are similar in that they try to balance covariate distributions.
However, because treatment has not yet been assigned in blocking problems, such algorithms only need to consider overall size constraints.
To solve matching problems, we must be able to impose more detailed compositional constraints.

\subsection{Matching objective} \label{sec:objective}

The matching objective is based on summaries of pairwise distances between units.
Let $\dm: \mathbf{U}\times \mathbf{U} \rightarrow \mathbb{R}^+$ be a distance metric capturing similarity between any pair of units, where lower values indicate greater similarity.
A distance metric is any function that satisfies:
\begin{enumerate}
	\item (Non-negativity) $\forall i,j\in\mathbf{U}, \dm(i,j) \geq 0$,
	\item (Self-similarity) $\forall i\in\mathbf{U}, \dm(i,i) =  0$,
	\item (Symmetry) $\forall i,j\in\mathbf{U}, \dm(i,j) =  \dm(j,i)$,
	\item (Triangle inequality) $\forall i,j,\ell\in\mathbf{U}, \dm(i,j) \leq \dm(i,\ell) + \dm(\ell, j)$.
\end{enumerate}
All commonly-used similarity measures, such as absolute differences between propensity scores and Euclidean or Mahalanobis distances in a covariate space, satisfy these conditions.

The objective function used in conventional full matching is either a weighted mean of within-group distances between treated and control units \citep{Rosenbaum1991} or the sum of such distances \citep{Hansen2004}.
We will depart from this convention in two ways.
First, we will focus on the maximum within-group distance, a \emph{bottleneck} objective function.
The main motivation for this shift is that the bottleneck objective facilitates the computationally efficient algorithm we present below.
However, while we only prove approximate optimality with respect to the maximum distance, the simulation study indicates that the algorithm performs well also with respect to the mean distances.
Apart from computational considerations, minimizing the maximum distance has the advantage of avoiding devastatingly poor matches that might be undetected by, for example, the mean distance \citep{Rosenbaum2017}.

The second departure is that we consider all within-group distances, not only those between units assigned to different treatment conditions as in the existing literature.
In the conventional full matching setting, there is little difference between the two objectives.
However, with more than two treatment conditions and larger matched groups, the conventional objective risks ignoring important within-group distances.
To maintain consistency with the current literature, we will also investigate the bottleneck objective that only includes within-group distances between units assigned to different conditions.
In symbols, the objectives we consider are:
\begin{eqnarray}
L^{Max}(\mathbf{M}) &=& \max_{\mathbf{m}\in\mathbf{M}}\max\{\dm(i,j) : i,j\in\mathbf{m}\},
\\
L^{Max}_{tc}(\mathbf{M}) &=& \max_{\mathbf{m}\in\mathbf{M}}\max\{\dm(i,j) : i,j\in\mathbf{m}\wedge W_i\neq W_j\}.
\end{eqnarray}

\subsection{Preliminaries}

The description of the algorithm and the proofs of its properties rely heavily on graph theory.
The most central concepts are defined in this section.
Refer to the supplementary materials for a brief overview of additional concepts and terminology.

\begin{definition}
	A \emph{closed neighborhood} of vertex $i$ in digraph $G=(V,E)$ is a subset of vertices $\N[i] \subset V$ consisting of $i$ itself and all vertices $j\in V$ with an arc from $i$ to $j$:
	\begin{equation}
	\N[i] = \{j\in V: (i,j)\in E\} \cup \{i\}.
	\end{equation}
\end{definition}

\begin{definition} \label{IJdigraph}
	An $\mathbf{IJ}$-digraph, denoted $G(\mathbf{I} \rightarrow \mathbf{J})$, is a graph $G = (\mathbf{I} \cup \mathbf{J}, E_{\mathbf{IJ}})$ with arcs drawn from all vertices in $\mathbf{I}$ to all vertices in $\mathbf{J}$. That is:
	\begin{equation*}
	E_{\mathbf{IJ}} = \{(i,j) : i \in \mathbf{I} \wedge j \in \mathbf{J} \},
	\end{equation*}
	Self-loops (i.e., arcs from $i$ to $i$) are drawn for all vertices $i\in\mathbf{I}\cap \mathbf{J}$.
\end{definition}

\begin{definition} \label{kappaNN}
	A \emph{$\kappa$-nearest neighbor digraph} of $G=(V,E)$ is a spanning subgraph of $G$ where an arc $(i,j)\in E$ is in the nearest neighbor digraph if $j$ is one of the $\kappa$ closest vertices to $i$ according to $\dm(i,j)$ among all its outward-pointing arcs.
	That is, for each $i\in V$, sort $(i,j)\in E$ by $\dm(i,j)$ and keep the $\kappa$ smallest arcs.
	If ties exist, give priority to self-loops and otherwise resolve them arbitrarily.
	We denote $\kappa$-nearest neighbor graphs as $\NN(\kappa, G)$.
\end{definition}

\subsection{The algorithm}

The following steps describe how a matching is constructed given a sample $\mathbf{U}$, matching constraints $\mathcal{C}=(c_1, \cdots, c_k, t)$ and distance metric $\dm(i,j)$.
Figure \ref{fig:algorithm} provides an illustration.

\begin{enumerate}
	\item For each treatment condition $x\in\{1, 2, \cdots, k\}$, construct the $c_x$-nearest neighbor digraph of the $\mathbf{U}\mathbf{w}_x$-digraph.
	Construct the union of these graphs:
	\begin{equation*}
	G_w = \NN(c_1, G(\mathbf{U}\rightarrow \mathbf{w}_1)) \cup \cdots \cup \NN(c_k, G(\mathbf{U} \rightarrow \mathbf{w}_k)).
	\end{equation*}

	\item Let $r = t - c_1 - \cdots - c_k$ be the number of units needed to satisfy the overall size constraint in excess of the treatment-specific constraints.
	Construct digraph $G_r$ by drawing an arc from $i$ to each of its $r$ nearest neighbors (of any treatment status) given that this arc does not exist in $G_w$:
	\begin{equation*}
	G_r = \NN(r, G(\mathbf{U} \rightarrow \mathbf{U}) - G_w),
	\end{equation*}
	where $G(\mathbf{U}\rightarrow \mathbf{U})$ is the complete digraph over $\mathbf{U}$ and the graph difference $G(\mathbf{U}\rightarrow \mathbf{U}) - G_w$ removes all arcs in $G(\mathbf{U}\rightarrow \mathbf{U})$ that exist in $G_w$.

	We refer to the union $G_\mathcal{C} = G_w \cup G_r$ as the \emph{$\mathcal{C}$-compatible nearest neighbor digraph}.

	\item Find a set of vertices $\mathbf{S}\subset \mathbf{U}$, referred to as \emph{seeds}, so that their closed neighborhoods in $G_\mathcal{C}$ are non-overlapping and maximal (i.e., adding any additional vertex to $\mathbf{S}$ would create some overlap).
	That is, $\mathbf{S}$ has the following two properties with respect to $G_\mathcal{C}$:
	\begin{itemize}
		\item (Independence) $\forall i,j\in \mathbf{S}, \N[i]\cap \N[j] = \emptyset$.
		\item (Maximality) $\forall j \not\in \mathbf{S}, \exists i\in \mathbf{S}, \N[i]\cap \N[j] \neq \emptyset$.
	\end{itemize}

	\item Assign a label to each seed. Assign the same label to all vertices in the seed's neighborhood in $G_\mathcal{C}$.
	We refer to vertices that are labeled in this step as \emph{labeled vertices}.

	\item For each vertex $i$ without a label, find its closed neighborhood $N[i]$ in $G_\mathcal{C}$ and assign it the same label as one of the labeled vertices in the neighborhood.
\end{enumerate}
When the algorithm terminates, each vertex has been assigned a label.
Vertices that share the same label form matched groups.
The collection of labels thus forms a matching.
Let $\mathbf{M}_{alg}$ denote this matching.

\begin{sidewaysfigure}
	\centering

	\includegraphics[width=0.95\textwidth]{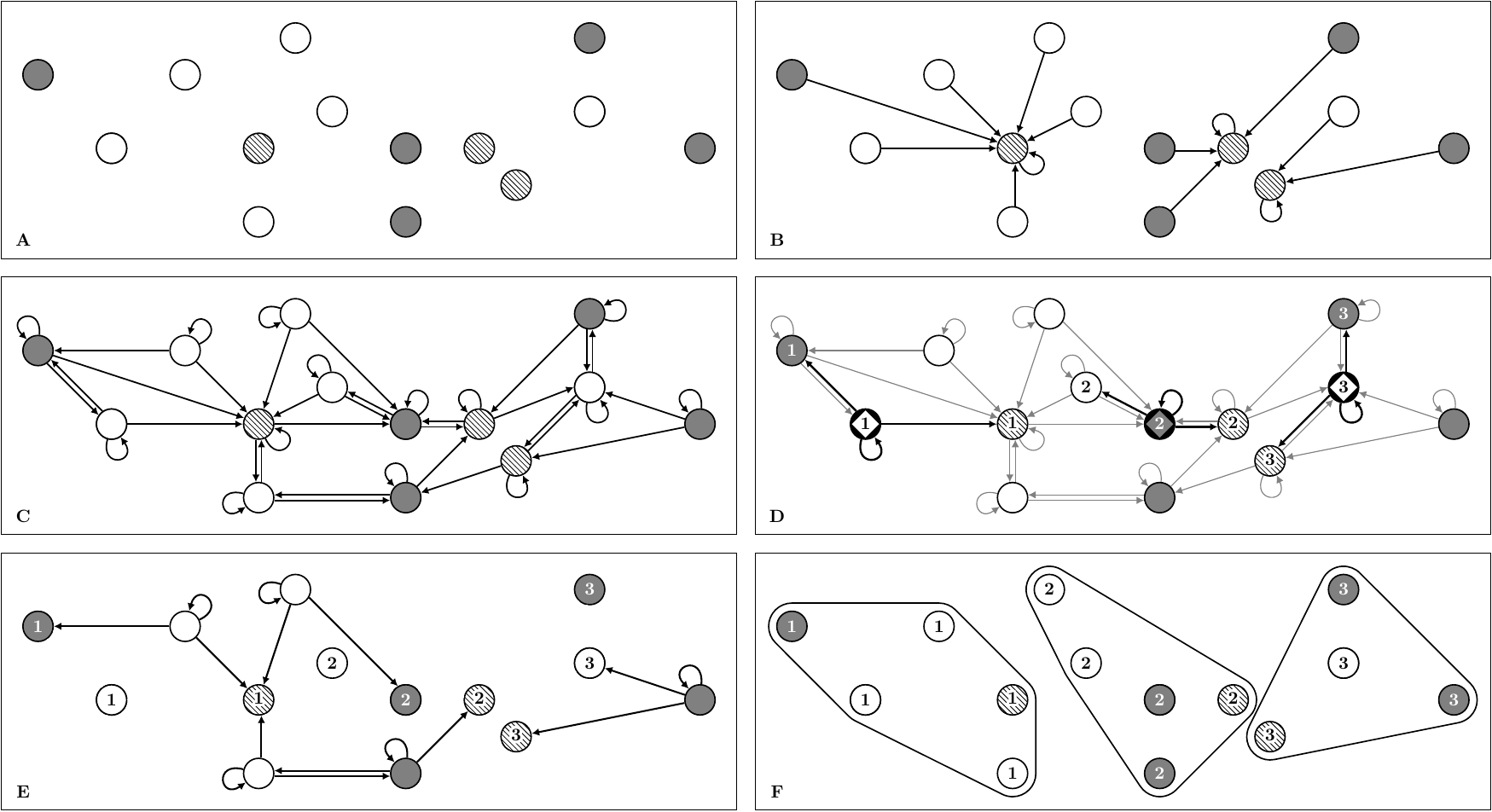}

	\vspace{0.1in}
	\caption{\scriptsize%
	The generalized full matching algorithm. The sample in this example consists of 14 units divided between three treatment conditions.
	We require that each matched group contains at least one unit of each treatment condition and at least three units in total.
	The distance metric is the Euclidean distance based on two covariates.
	\textbf{(A)} The units are presented as circles on the covariate plane.
	The treatment conditions are indicated by the units' color and pattern.
	\textbf{(B)} Step 1: One of the building blocks of $G_w$ is shown, namely the nearest neighbor digraph between the whole sample and the patterned units: $\NN(1, G(\mathbf{U}\rightarrow \mathbf{w}_\text{patterned}))$.
	\textbf{(C)} Step 2: The $\mathcal{C}$-compatible nearest neighbor digraph, $G_\mathcal{C}$, is created.
	Note that all vertices in this graph are pointing to one vertex of each treatment condition and that there exists no other graph with shorter arcs that satisfy this condition.
	\textbf{(D)} Step 3: A set of seeds is found.
	Seeds are indicated with a diamond shape enclosed in their circles.
	The arcs pointing out from the seeds are highlighted.
	Note that no two seeds are pointing to a common unit.
	Step 4: Each seed and its neighbors are given a unique label as indicated by the numbers.
	\textbf{(E)} Step 5: Some units are still unlabeled.
	Each such unit is assigned a label that is represented in its neighborhood.
	All outward-pointing arcs from unlabeled units are shown in this panel.
	\textbf{(F)} The algorithm has terminated. Matched groups are formed by units sharing the same label.
	} \label{fig:algorithm}
\end{sidewaysfigure}

\subsection{Properties}

The algorithm and the matching it constructs have two key properties.
First, $\mathbf{M}_{alg}$ is a 4-approximate generalized full matching.
That is, it is an admissible generalized full matching, and the maximum within-group distance in the matching is guaranteed to be less or equal to four times the maximum within-group distance in an optimal matching.
Second, the algorithm terminates quickly.
In this section, we discuss these two properties in detail.
Formal proofs are presented in the supplementary materials.

\subsubsection{Optimality}

Approximate optimality follows from two properties of the $\mathcal{C}$-compatible nearest neighbor digraph, as shown by the following two lemmas.

\begin{lemma} \label{prop:closedN-admissible}
	The closed neighborhood of each vertex in the $\mathcal{C}$-compatible nearest neighbor digraph satisfies the matching constraints $\mathcal{C}=(c_1, \cdots, c_k, t)$:
	\begin{equation}
	\forall i\in V,\, \forall x \in \{1, \cdots, k\},\, |\N[i] \cap \mathbf{w}_x| \geq c_x, \qquad \text{and} \qquad \forall i\in V,\, |\N[i]| \geq t.
	\end{equation}
\end{lemma}

\begin{lemma} \label{prop:NNGbound}
	The distance between any two vertices connected by an arc in the $\mathcal{C}$-compatible nearest neighbor digraph, $G_\mathcal{C} = (V, E_\mathcal{C})$, is less or equal to the maximum within-group distance in an optimal matching:
	\begin{equation}
	\forall (i,j)\in E_\mathcal{C},\; \dm(i,j) \leq \min_{\mathbf{M}\in\mathcal{M}_\mathcal{C}} L^{Max}(\mathbf{M}).
	\end{equation}
\end{lemma}

Lemma \ref{prop:closedN-admissible} states that the $\mathcal{C}$-compatible nearest neighbor digraph encodes the matching constraints in the units' neighborhoods in $G_\mathcal{C}$.
Admissibility of $\mathbf{M}_{alg}$ follows from that each matched group is a superset of a closed neighborhood.
Since each neighborhood satisfy the matching constraints, so will each group.

Lemma \ref{prop:NNGbound} provides a connection between the arc weights in the $\mathcal{C}$-compatible nearest neighbor digraph and the maximum distance in the optimal solution.
This follows from that a digraph compatible with $\mathcal{C}$ (in the sense that it satisfies Lemma \ref{prop:closedN-admissible}) can be constructed as a subgraph of the cluster graph induced by an optimal matching.
A $\mathcal{C}$-compatible digraph constructed in this way satisfies the property in Lemma \ref{prop:NNGbound} because its construction does not require the addition of any arcs not already in the optimal matching.
We show in the supplementary materials that $G_\mathcal{C}$ is the smallest digraph compatible with $\mathcal{C}$.
Consequently, the distances in $G_\mathcal{C}$ must be bounded in the same way as in the subgraph induced by an optimal matching, and Lemma \ref{prop:NNGbound} follows.

Approximate optimality follows from the triangle inequality.
In particular, as shown in supplementary materials, two units in the same matched group is at most at a geodesic distance of four arcs in the $\mathcal{C}$-compatible nearest neighbor digraph.
The worst case is when the five vertices connected by the arcs are lined up on a straight line in the metric space.
In that case, the distance between the two end vertices is the sum of the distances of the intermediate arcs.
Lemma \ref{prop:NNGbound} provides a bound for the intermediate arc distances and, thus, a bound for all within-group distances.

\begin{theorem} \label{prop:approx-opt}
	$\mathbf{M}_{alg}$ is a 4-approximate generalized full matching with respect to the matching constraint $\mathcal{C}=(c_1, \cdots, c_k, t)$ and matching objective $L^{Max}$:
	\begin{equation}
	\mathbf{M}_{alg}\in \mathcal{M}_\mathcal{C}, \qquad \text{and} \qquad L^{Max}(\mathbf{M}_{alg}) \leq \min_{\mathbf{M}\in\mathcal{M}_\mathcal{C}} 4 L^{Max}(\mathbf{M}).
	\end{equation}
\end{theorem}

A similar strategy can be used to bound the $L^{Max}_{tc}$ objective for the matching produced by the algorithm.
In particular, Lemma \ref{prop:NNGbound} holds also for the bottleneck objective function for distances between treated and controls for conventional full matching problems, as described by Definition \ref{def:otfm}.

\begin{theorem} \label{prop:approx-opt-trad}
	$\mathbf{M}_{alg}$ is a 4-approximate conventional full matching with respect to the matching constraint $\mathcal{C}=(1, \cdots, 1, k)$ and matching objective $L^{Max}_{tc}$:
	\begin{equation}
	\mathbf{M}_{alg}\in \mathcal{M}_\mathcal{C}, \qquad \text{and} \qquad L^{Max}_{tc}(\mathbf{M}_{alg}) \leq \min_{\mathbf{M}\in\mathcal{M}_\mathcal{C}} 4 L^{Max}_{tc}(\mathbf{M}).
	\end{equation}
\end{theorem}

\subsubsection{Complexity}

Following the convention in the matching literature \citep{Abadie2006}, we assume that the number of treatment conditions and the matching constraints are fixed asymptotically in the following theorem.
The time complexity is still polynomial if we let the number of treatment conditions and the matching constraints grow proportionally in $n$, but the exposition becomes less clear.

\begin{theorem} \label{prop:complexity}
	In the worst case, the generalized full matching algorithm terminates in polynomial time using linear memory.
\end{theorem}

The algorithm can be divided into two parts.
The first and more intricate part is the construction of the $\mathcal{C}$-compatible nearest neighbor digraph.
This essentially acts as a preprocessing step of the remainder of the algorithm.
The idea is that $G_\mathcal{C}$ encodes sufficient information about the sample to ensure approximate optimality, but that it is sparse enough to ensure quick execution.
Once $G_\mathcal{C}$ is constructed, the remaining steps are completed in linear time.

As discussed in the proof of Theorem \ref{prop:complexity} in the supplementary materials, the $\mathcal{C}$-compatible nearest neighbor digraph can be constructed by $O(n)$ calls to a nearest neighbor search function.
For an arbitrary metric, each such call has a time complexity of $O(n \log n)$ and a space complexity of $O(n)$.
It follows that the overall worst-case time complexity is $O(n^2 \log n)$.

Specialized nearest neighbor search algorithms exist for the most commonly-used distance metrics.
For example, when the metric is the Euclidean or Mahalanobis distance in a covariate space, large improvements can be expected by storing the data points in a kd-tree \citep{Friedman1977}.
Given that the covariate distribution is not too skewed, each search can then be completed in $O(\log n)$ time on average.
Using this approach, the overall average time complexity would be reduced to $O(n \log n)$.
However, these data structures do not scale well with the dimensionality of the metric space.
Feasible approaches in such cases include reducing the dimensionality prior to matching (e.g., by matching on estimated propensity scores, \citealp{Rosenbaum1983}), repeated matching in random low-dimensional projections of the covariate space \citep{Li2016} and using approximate nearest neighbor search algorithms \citep{Arya1998}.
Nearest neighbor search can be completed in constant time when the distances are absolute differences in propensity scores.

\subsection{Extensions} \label{sec:extensions}

The algorithm admits several extensions and refinements.
First, the set of seeds derived in the third step of the algorithm is not unique.
The properties discussed above hold for any set of seeds, but, heuristically, the performance of the matching depends on the units that are picked.
A valid set of seeds are equivalent to a maximal independent vertex set in the graph described by the adjacency matrix $\mathbf{A}\mathbf{A}' + \mathbf{A} + \mathbf{A}'$ where $\mathbf{A}$ is the adjacency matrix of $G_\mathcal{C}$. We expect improvements if a larger maximal independent set is used as seeds.

Second, in the fifth step of the algorithm, unassigned vertices are assigned to groups based on the $\mathcal{C}$-compatible nearest neighbor digraph.
However, as all matching constraints have already been fulfilled in the fourth step, the restrictions encoded in $G_\mathcal{C}$ are no longer necessary.
By restricting the matches to arcs in $G_\mathcal{C}$, we might miss matched groups that are closer to the unassigned units.
We could improve quality by searching for the closest labeled vertex among all vertices.

Third, it is sometime beneficial to relax the restriction that all units must be assigned to a matched group.
For example, if some regions of the covariate space is sparse with respect to a treatment condition, we could be forced to construct groups of poor quality in order to avoid discarding units.
A common way to avoid this is to apply a \emph{caliper}.
That is, we restrict the maximum allowable distance within any matched group to some value.
Units that cannot be assigned a group without violating the caliper are discarded.
In our algorithm, such a caliper can be imposed in the construction of $G_\mathcal{C}$.
In particular, by restricting the length of the arcs in $G_\mathcal{C}$, we implicitly restrict the maximum allowable distance in the resulting matching.
If the second refinement is implemented, we can impose a caliper (perhaps of different magnitude) also when assigning units in the fifth step.

Fourth, we are occasionally interested in estimating treatment effects only for some subpopulation.
For example, it is common to estimate the average treatment effect only for treated units.
To estimate such an effect, we only need that units from the subpopulation of interest are assigned to matched groups; other units can be left unassigned.
In most cases, it can still be beneficial to assign all units to groups as we then exploit all information in the sample.
However, if there are sparse regions in the covariate space, such a procedure might lead to poor match quality.
The algorithm allows us to focus the matching to a certain set of units.
In particular, in the first two steps, by substituting $\mathbf{U}$ with some subset $\mathbf{B} \subset \mathbf{U}$, we ensure that all units in $\mathbf{B}$ are assigned to matched groups.
Units not in $\mathbf{B}$ are only assigned to groups insofar as they are needed to satisfy the matching constraints.
The unassigned units may later be assigned to groups in the fifth step (preferably with a caliper to avoid affecting match quality).

\section{Simulation study} \label{sec:simulations}

We present the results from a small simulation study of the performance of an implementation of the presented algorithm.
The comparison is with conventional full matching and nearest neighbor matching with and without replacement.
We include both optimal and heuristic (``greedy'') implementations of nearest neighbor matching.

We focus on a simple setting where each unit has two covariates distributed uniformly on a plane:
\begin{equation}
X_{1i}, X_{2i} \sim \mathcal{U}(-1,1).
\end{equation}
To facilitate the comparison with previous methods, the units are randomly assigned to one of two treatment conditions, $W_i\in\{0,1\}$, based on a logistic propensity score that maps from the covariates to treatment probabilities:
\begin{equation}
\Pr(W_i = 1|X_{1i}, X_{2i}) = \logistic\left[\frac{(X_{1i} + 1)^2 + (X_{2i} + 1)^2 - 5}{2}\right].
\end{equation}
Units with larger covariate values are thus more likely to be treated.
The treatment probability ranges from 7.6\% at $(-1,-1)$ to 81.8\% at (1,1).
The unconditional treatment probability is 26.5\%. The outcome is given by:
\begin{equation}
Y_i = (X_{1i} - 1)^2 + (X_{2i} - 1)^2 + \varepsilon,
\end{equation}
where $\varepsilon$ is standard normal.
The outcome does not depend on the assigned treatments, so the treatment effect is constant at zero.

We use a version of the estimator discussed by \citet{Abadie2006} and \citet{Imbens2015} to estimate the average treatment effect for the subpopulation of treated units (\textsc{att}).
We derive the mean outcome difference between treated and control units within each matched group and then average the differences over all groups weighted by the number of treated units:
\begin{equation}
\widehat{\tau_\textsc{att}}(\mathbf{M}) = \sum_{\mathbf{m} \in \mathbf{M}} \frac{|\mathbf{w}_1\cap\mathbf{m}|}{|\mathbf{w}_1|} \left[\frac{\sum_{i\in\mathbf{m}} W_i Y_i}{|\mathbf{w}_1\cap\mathbf{m}|} - \frac{\sum_{i\in\mathbf{m}} (1-W_i) Y_i}{|\mathbf{w}_0\cap\mathbf{m}|}\right].
\end{equation}
Other estimation approaches are discussed by \citet{Stuart2010}.
In particular, matching facilitates permutation-based inference which many investigators find useful \citep{Rosenbaum2002,Rosenbaum2010}.

The Savio cluster at UC Berkeley was used to run the simulations using version 3.3.2 of R.
Each simulation round was assigned a single \textsc{cpu} core, largely reflecting the performance of a modern computer.
To derive generalized full matchings, we used a development version of the \texttt{quickmatch} R package.
Optimal conventional full matchings and optimal nearest neighbor matchings were derived using version 0.9-7 of the \texttt{optmatch} R package \citep{Hansen2006}.
Version 4.9-2 of the \texttt{Matching} R package \citep{Sekhon2011} was used to derive greedy matchings and matchings with replacement.
The conventional and generalized full matching methods use the same matching constraints, namely at least one treated and control unit in each group.
We used the Euclidean distances on the covariate plane as the similarity measure in all cases.
The \texttt{quickmatch} package uses the maximum within-group distance as its objective function, as discussed above.
The \texttt{Matching} and \texttt{optmatch} packages use the sum of within-group distances between treated and control units as their objectives.
All functionality beside the matching functions (e.g., estimators) was implemented independently and is common for all matching methods.

\subsection{Run time and memory}

We matched 1,000 randomly generated samples with each matching method for sample sizes ranging from 100 units to 100 million units.
Figure \ref{fig:complexity} presents the computational resources used by each implementation as a function of sample size.
Average runtimes are presented in the first three panels, and memory usage is presented in the subsequent three panels.
The results are split into several panels with different scales due to the large differences in performance.
Table \ref{tab:complexity} in the supplementary materials provides additional details.

\begin{sidewaysfigure}
	\centering

	\includegraphics[width=0.28\textwidth]{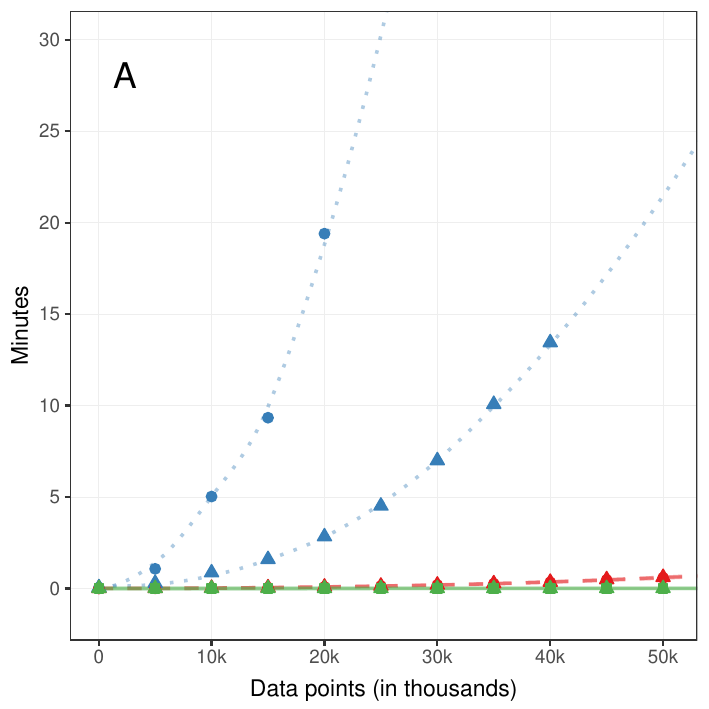}
	\hspace{0.075in}
	\includegraphics[width=0.28\textwidth]{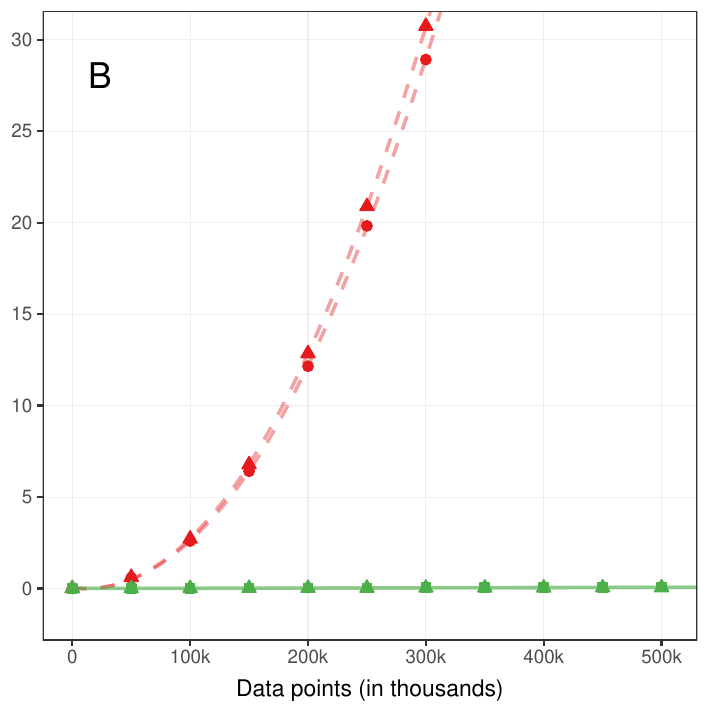}
	\hspace{0.075in}
	\includegraphics[width=0.28\textwidth]{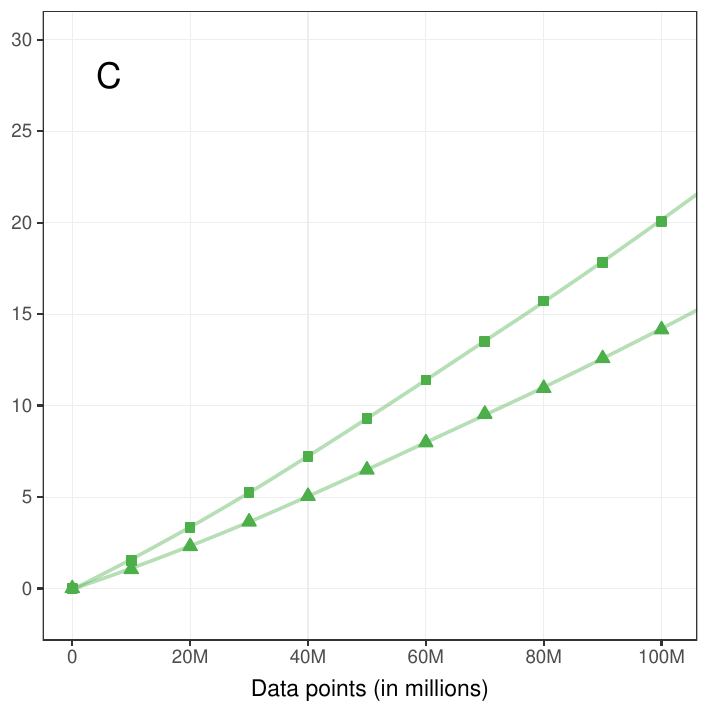}
	\vspace{0.05in}

	{\scriptsize Panels A, B \& C: Runtime (in minutes)} \\

	\vspace{0.3in}
	\includegraphics[width=0.28\textwidth]{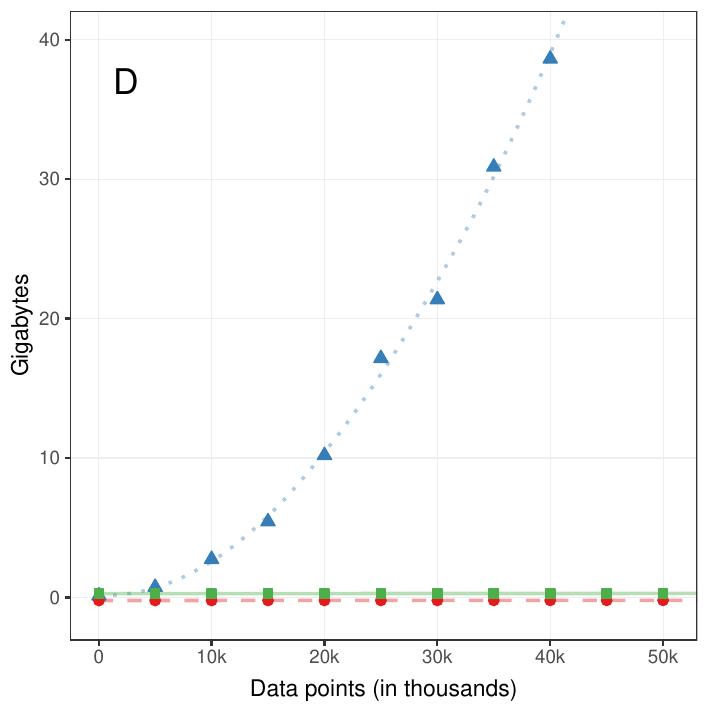}
	\hspace{0.075in}
	\includegraphics[width=0.28\textwidth]{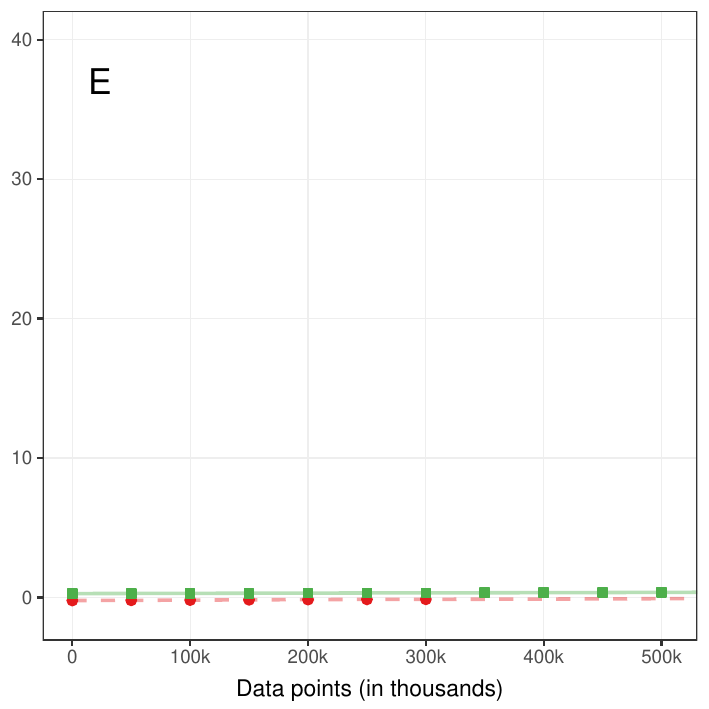}
	\hspace{0.075in}
	\includegraphics[width=0.28\textwidth]{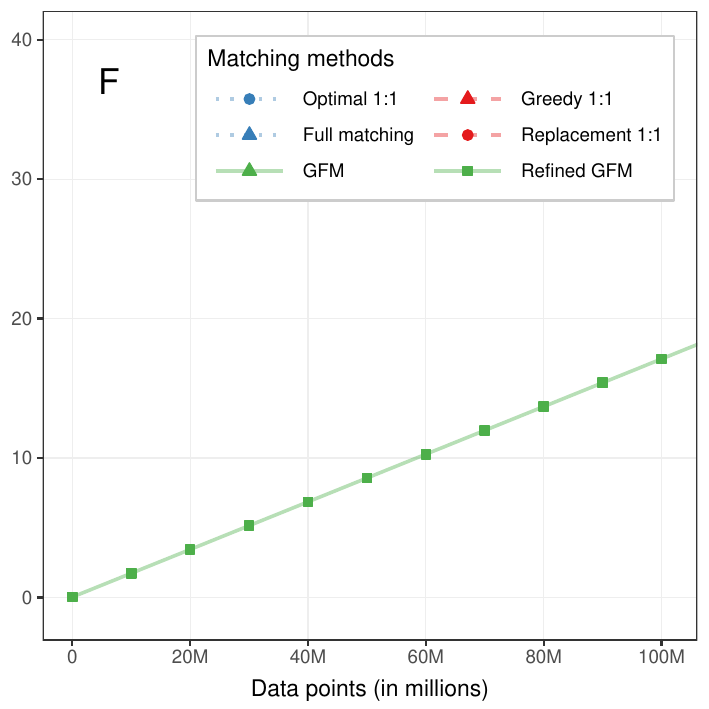}
	\vspace{0.05in}

	{\scriptsize Panels D, E \& F: Memory (in gigabytes)} \\

	\vspace{0.2in}
	\caption{\scriptsize%
	Runtime and memory use by matching method. Marker symbols are actual simulation results, and connecting lines are interpolations.
	Memory use was identical for methods from the same package, so we present results for only one implementation from each package.
	Each measure is based on 1,000 simulation rounds.
	The simulation errors are negligible.
	} \label{fig:complexity}
\end{sidewaysfigure}

Panels A and D present results for samples with up to 50,000 units.
For small sample sizes, all implementations perform well.
However, as the sample grows, the \texttt{optmatch} package struggles both with respect to runtime and memory.
Already with 10,000 units, optimal nearest neighbor matching takes more than 25 minutes to terminate on average, and with sample sizes over 40,000 units, the package allocates more than 40 gigabytes of memory.
The implementations in \texttt{optmatch} are the only ones that derives optimal solutions, but this comes at a large computational cost.
The other packages terminate close to instantly with negligible memory use for these sample sizes.

Results for samples with up to 500,000 units are presented in Panels B and E.
Implementations from the \texttt{quickmatch} package still terminates virtually instantly with negligible memory use.
The \texttt{Matching} package terminates quickly for samples with less than 200,000 units, but its runtime increases after that.
More than 30 minutes are required for samples larger than about 300,000 units. Memory use is, however, still negligible.

Panels C and F present samples with up to 100 million units.
Both implementations of the generalized full matching algorithm terminate quickly for sample sizes less than 20 million units.
With a sample of 100 million units, the implementation without refinements terminates within 15 minutes on average, while the version with refinements adds about 5 minutes to the runtime.
Memory use increases at a slow, linear rate. At 20 million units, it uses about 4 gigabytes of memory on average.
At 100 million units, it requires slightly more than 17 gigabytes.

\subsection{Match quality}

To investigate the quality of the matched groups, we matched 10,000 randomly generated samples containing 1,000 and 10,000 units.
The results differ little with the sample size, so the results for samples with 1,000 units are presented in the supplementary materials.
We also restrict our focus to performance with respect to covariate balance and the behavior of the effect estimator.
The performance with respect to group structure and aggregated distances is presented in the supplementary materials.

The first five columns in Table~\ref{tab:performance} present the absolute mean difference between the adjusted treatment groups on the first two moments of the covariates.
The adjustment used to assess covariate balance is the same as for the estimator.
We include the balance in the unadjusted sample before matching for comparison.
The scaling of the balance is arbitrary, so the results are normalized by the results for conventional full matching to ease interpretation.
How well the methods balance the samples depends on the data generating process.
If, for example, the propensity score is constant, matching would only correct chance imbalances due to sampling variation and, thus, only lead to minor improvements compared to the unadjusted sample.
We do not know how representative the simulation study is of the methods performance in general, but we have no reason to believe the qualitative conclusions would change.

All matching adjustments yield large improvements in covariate balance compared to the unadjusted sample.
The one exception is the cross-moment of the covariates for the two implementations of nearest neighbor matching without replacement, where the balance is slightly worse than when no adjustment is done.
This is largely an effect of that those moments are fairly balanced already in the unadjusted sample, but the behavior is still disconcerting.
The four remaining implementations lead to large improvements on all moments.
Nearest neighbor matching with replacement produces the greatest balance with the full matching methods as a close second.
The former method achieves this by dropping $54.7\%$ of the sample (see Table~\ref{tab:group-structure} in supplementary materials).
The full matching methods do not drop any units, with only a minor decrease in balance.
We note that generalized full matching without refinements leads to better balance than the other full matching methods.
We have not found an explanation for this behavior and do not expect it to hold across settings.

\begin{table}
	\centering
	\caption{Performance of matching methods with samples of 10,000 units.} \label{tab:performance}
	\resizebox{\textwidth}{!}{
	\begin{tabular}{l c rrrrr p{0.1in} rrrr}
		& & \multicolumn{5}{c}{Covariate balance} & & \multicolumn{4}{c}{Estimator performance} \\ \cline{3-7} \cline{9-12}
		& & \multicolumn{1}{r}{$X_1$} & \multicolumn{1}{r}{$X_2$} & \multicolumn{1}{r}{$X_1^2$} & \multicolumn{1}{r}{$X_2^2$} & \multicolumn{1}{r}{$X_1X_2$} & & \multicolumn{1}{r}{Bias} & \multicolumn{1}{r}{\textsc{se}} & \multicolumn{1}{r}{\textsc{rmse}} & \multicolumn{1}{r}{$\frac{\text{Bias}}{\text{\textsc{rmse}}}$} \\ \hline
		Unadjusted      & &                        500.32 &                        502.00 &                        101.07 &                        100.69 &                        131.37 & &                                                             1087.12 &                                                                1.48 &                                                               39.81 &                                                               0.999  \\
Greedy 1:1      & &                         50.19 &                         50.32 &                         62.74 &                         62.56 &                        139.13 & &                                                               53.68 &                                                                1.04 &                                                                2.22 &                                                               0.884  \\
Optimal 1:1     & &                         50.10 &                         50.25 &                         62.24 &                         62.07 &                        139.80 & &                                                               54.14 &                                                                1.04 &                                                                2.24 &                                                               0.885  \\
Replacement 1:1 & &                          0.41 &                          0.41 &                          0.73 &                          0.73 &                          0.80 & &                                                                0.32 &                                                                1.17 &                                                                1.17 &                                                               0.010  \\
Full matching   & &                          1.00 &                          1.00 &                          1.00 &                          1.00 &                          1.00 & &                                                                1.00 &                                                                1.00 &                                                                1.00 &                                                               0.037  \\
GFM             & &                          0.71 &                          0.71 &                          0.71 &                          0.72 &                          0.76 & &                                                                0.57 &                                                                1.03 &                                                                1.03 &                                                               0.020  \\
Refined GFM     & &                          1.03 &                          1.03 &                          0.97 &                          0.97 &                          1.09 & &                                                                1.19 &                                                                1.01 &                                                                1.01 &                                                               0.043  \\
 \hline
	\end{tabular}
	}
\end{table}

We continue with an investigation of how the matching methods affect the treatment effect estimator.
As with the balance measures, these results depend on the details of the data generating process.
While the details might not generalize, the qualitative conclusions should.
The final four columns in Table~\ref{tab:performance} present the results.

The first column presents the bias of the estimator.
As expected, the estimator has substantial bias in the unadjusted sample.
Nearest neighbor matching without replacement leads to a reduction of about 95\%.
While this is a substantial improvement, it is still more than an order of magnitude greater than the bias for nearest neighbor matching with replacement and full matching.

The second column reports the estimator's standard error.
The standard error depends mainly on two factors.
First, while matching is primarily used to adjust for systematic covariate differences between the treatment groups, it will also adjust for unsystematic differences.
Such chance-imbalances lead to increased estimator variance, and matching may, thus, improve precision.
Second, for a given level of balance, larger variation in the weights induced by the matching will lead to a higher standard error since the information in the sample is used less effectively.
The trade-off between balance and weight variability is reflected in the standard errors.
It is particularly evident when the standard errors for nearest neighbor matching with and without replacement are compared.
Matching with replacement induces a larger variation in the matching weights, and the standard error is $12.5\%$ larger than with matching without replacement.
The lowest variance is, however, given by the full matching methods.

The final two columns investigate the root mean square error (\textsc{rmse}) of the estimator and the bias's share of the \textsc{rmse}.
The full matching methods lead to both low bias and variance, so they perform well on the \textsc{rmse}.
Matching with replacement produces a 17\% larger \textsc{rmse} compared to conventional full matching, but it is still considerably lower than matching without replacement.
The bias to \textsc{rmse} ratio shows how accurate our inferences will be.
In particular, if the systematic error is a large part of the total error, variance estimators will not capture the uncertainty of the estimation and our conclusions will be misleading.
This measure is scale-free and is, therefore, not normalized.
As expected, the \textsc{rmse} consists almost exclusively of bias in the unadjusted sample.
Any conclusions from such analyses are likely misleading.
Matching without replacement only produces minor improvements.
In stark contrast, matching with replacement and full matching lead to large reductions in the ratio; the bias is only between 1\% and 4\% of the \textsc{rmse}.
This ratio critically depends on the dimensionality of the covariate space \citep{Abadie2006}, but we expect a similar pattern to hold across settings.

\section{Extrapolation from a voter mobilization experiment}\label{sec:application}

We return to the voter mobilization experiment.
The objective is to extrapolate the results to the overall population in the 2006 primary election in Michigan.
The experimental sample was constructed from Michigan's Qualified Voter File (\textsc{qvf}).
The Bureau of Elections in Michigan created the \textsc{qvf} in 2002 in an effort to modernize their highly decentralized voter registration system, and by 2006 the file contained 6,762,701 registered voters.\footnote{The process of transferring the voter registration system to the \textsc{qvf} was not complete by 2006, and a small portion ($5.8\%$) of the electorate is missing from the data set. This explains the difference between the $17.7\%$ turnout rate cited in the introduction and the rates presented in Table~\ref{tab:result-main}.}
We ask what the turnout would have been if the treatments in the experiment were implemented in the complete population of registered voters in the \textsc{qvf}.

The treatment arm of main interest in \citet{Gerber2008Social} was the postcard with the voting history of the recipient's neighborhood (the ``Neighbors'' condition).
The authors were, however, worried that the postcards could affect voting behavior through other channels than social pressure, and the additional treatment arms were added to shine light on this.
The first concern was that the postcard could simply remind the recipient of the upcoming election.
A condition was added (``Civic Duty'') with a postcard stating that it was the recipient's civic duty to vote in the upcoming election without information about voting history.
If social pressure was an important determinant of voting in this sample, we would expect there to be a noticeable difference in turnout between the Neighbors and Civic Duty conditions.
A second concern was the so called Hawthorne effect, namely that knowledge that one is being studied can itself affect behavior.
A third condition was added (``Hawthorne'') with a postcard stating that the authors would be studying voting behavior in the election, and particularly observe the recipient's voting decision through public records.
The final concern was that being reminded about one's own voting history might affect behavior irrespectively of knowledge about the voting pattern of one's neighbors.
The fourth condition (``Self'') was a postcard listing the voting history of the recipient without any information about their neighbors.
The final condition (``Control'') was a pure control group where the registered voters did not receive a postcard.

\begin{table}
	\centering
	\caption{Unadjusted and matching adjusted average turnout in the 2006 primary election.} \label{tab:result-main}
	\resizebox{\textwidth}{!}{
	\begin{tabular}{l rrrrrr}
		& \multicolumn{1}{r}{Control} & \multicolumn{1}{r}{Civic Duty} & \multicolumn{1}{r}{Hawthorne} & \multicolumn{1}{r}{Self} & \multicolumn{1}{r}{Neighbors} & \multicolumn{1}{r}{Non-experiment} \\ \hline
		Unadjusted turnout (\%) &   29.66 &   31.45 &   32.24 &   34.52 &   37.79 &   18.01 \\
Adjusted turnout (\%) &   21.43 &   23.73 &   23.01 &   25.16 &   26.88 &   18.60 \\
Observations &   191,243 &    38,204 &    38,218 &    38,201 &    38,218 & 6,418,617 \\
 \hline
	\end{tabular}
	}
\end{table}

The first row in Table~\ref{tab:result-main} presents the average turnout within each of the five treatment conditions.
We see that the Neighbors condition led to the largest turnout of $37.8\%$, but the three other postcard conditions still increased turnout compared to control.
Given the large size of the experiment, the standard errors are negligible compared to the treatment effects and will not concern us in this discussion.
The final column in Table~\ref{tab:result-main} presents the turnout among registered voters not included in the experiment.
People in this group was not sent a postcard, so their treatment is effectively the same as the control group in the experiment.
The turnout is, however, more than eleven percentage points higher in the control group than in the non-experimental group.
This is an indication of how highly selected experimental sample is.

To extrapolate the results from the experiment, we construct matched groups with all registered voters in the \textsc{qvf} so that each group contains at least one unit from each treatment condition.
The matching was done in R using the generalized full matching algorithm implemented in the \texttt{quickmatch} package, and it was completed within two minutes on a laptop computer.
We include all covariates discussed by the original authors in the matching: age measures in days, gender, and past voting history.
The voting history consists of indicators of whether the registered voter voted in the primary elections in August of 2000, 2002 and 2004, and in the partisan elections in November of 2000 and 2002.
The exclusion of the partisan election in 2004 is discussed below.
We also include geographical coordinates of the address of each registered voter.
Mahalanobis distances are used to measure similarity between voters.
Table~\ref{tab:balance} provides averages of all variables except the coordinates for the control condition and the non-experimental group before and after matching.
The supplementary material present unadjusted and adjusted covariate averages for all treatment conditions.
We note, as expected, large improvements in balance after matching adjustment, except for voting in the partisan election in 2004.

\begin{table}
	\centering
	\caption{Covariate balance with and without matching adjustment.} \label{tab:balance}
	\resizebox{\textwidth}{!}{
		\begin{tabular}{l c rr p{0.25in} rr}
		& & \multicolumn{2}{c}{Unadjusted} & & \multicolumn{2}{c}{Matching adjustment} \\ \cline{3-4} \cline{6-7}
		& & \multicolumn{1}{c}{Control} & \multicolumn{1}{c}{Non-experiment} & & \multicolumn{1}{c}{Control} & \multicolumn{1}{c}{Non-experiment} \\ \hline
		Birth year &  & 1956.19 & 1957.96 &  & 1958.16 & 1957.87 \\
Female (\%) &  &   49.89 &   53.32 &  &   53.29 &   53.15 \\
Voted Aug 2000 (\%) &  &   25.19 &   14.65 &  &   15.19 &   15.19 \\
Voted Aug 2002 (\%) &  &   38.94 &   22.59 &  &   23.42 &   23.43 \\
Voted Aug 2004 (\%) &  &   40.03 &   18.71 &  &   19.80 &   19.80 \\
Voted Nov 2000 (\%) &  &   84.34 &   52.49 &  &   54.11 &   54.11 \\
Voted Nov 2002 (\%) &  &   81.09 &   41.93 &  &   43.94 &   43.92 \\
Voted Nov 2004 (\%) &  &  100.00 &   67.57 &  &  100.00 &   68.76 \\
 \hline
	\end{tabular}
	}
\end{table}

The second row in Table~\ref{tab:result-main} presents the turnout of the six conditions after matching adjustment.
The figures should be interpreted as estimates of turnout of the six conditions if scaled up to the whole population.
That is, the turnout when all registered voters, both those in the experiment and those not, were exposed to the corresponding treatment.
We expect the estimates to be accurate representations of the counterfactual turnout if the matching was successful and the selection-on-observables assumption holds.
Indeed, we see that the turnout is lower for all treatment conditions compared to the experiment, reflecting that voters with a high baseline vote propensity were included in the experimental sample.

There is no direct way to test the selection-on-observables assumption, so the quality of the extrapolation can typically not be judged directly.
We can, however, do so in this case because the pure control group and the registered voter not included in the experiment received the same treatment.
The turnout should be essentially the same in the two groups if the matching adjustment was successful.
This is, however, not what we observe.
The turnout is almost three percentage points higher in the control group than in the non-experimental condition.
The failure of this placebo test is a strong indication that the extrapolation was unsuccessful.

We need to consider the voting history in the 2004 partisan election to understand this result.
The authors' sample selection was based on the proprietary indices discussed in the introduction, but they also required that all registered voters in the experimental sample had voted in the 2004 partisan election.
In contrast, only $67.6\%$ of the registered voters not included in the experiment voted in the election.
The consequence is that the support of the covariate distribution in the experiment does not overlap with covariate distribution in the population.
Subsequently, no adjustment exists to balance the distributions along this dimension.
The only thing that can salvage the validity of the matching estimates presented above is the assumption that voting behavior in the 2004 partisan election is independent of voting behavior in the 2006 primary election conditional on the remaining covariates.
This assumption is unlikely to hold.

A simple solution is to move the inferential target to the registered voter in the overall population who did vote in 2004 partisan election.
Overlap is ensured with respect to this subpopulation, so extrapolation can be successful without the strong assumptions that otherwise would have been necessary.
Of course, the effects in this subpopulation are likely different than the effects in the complete population, so the estimates may not provide an answer to the question of ultimate interest.
However, the available information limits what questions can be answered with credibility, and we must abide.

Table \ref{tab:result-g4} presents the turnout of the six conditions before and after adjustment for the subpopulation of voters in 2004 partisan election.
The unadjusted turnout in the non-experimental group increases compared to Table \ref{tab:result-main}.
These voters were, as we might expect, more likely to vote in the election in absence of any postcard.
There is, however, still a substantial difference between the non-experimental group and the control group in the experiment, indicating that further adjustments are required.
The second row in Table \ref{tab:result-g4} presents the results after adjustment using generalized full matching.
The difference between the control group and the non-experimental group is small but not zero.
This indicates that the adjustment is not perfect.
The remaining difference could, for example, be an indication that some additional information was used in sample selection that we are unaware of.
The placebo test can, however, be marked as a weak pass, and the results should therefore provide a reasonable, but not perfect, indication of the counterfactual turnout if the treatments were scaled up to this subpopulation.

\begin{table}
	\centering
	\caption{Turnout in the 2006 primary election among voters in the 2004 partisan election.} \label{tab:result-g4}
	\resizebox{\textwidth}{!}{
	\begin{tabular}{l rrrrrr}
		& \multicolumn{1}{r}{Control} & \multicolumn{1}{r}{Civic Duty} & \multicolumn{1}{r}{Hawthorne} & \multicolumn{1}{r}{Self} & \multicolumn{1}{r}{Neighbors} & \multicolumn{1}{r}{Non-experiment} \\ \hline
		Unadjusted turnout (\%) &   29.66 &   31.45 &   32.24 &   34.52 &   37.79 &   25.56 \\
Adjusted turnout (\%) &   26.59 &   28.86 &   27.95 &   30.87 &   32.90 &   25.89 \\
Observations &   191,243 &    38,204 &    38,218 &    38,201 &    38,218 & 4,337,193 \\
 \hline
	\end{tabular}
	}
\end{table}

The adjusted averages in Table \ref{tab:result-g4} suggests that the postcards would have increased turnout.
The Neighbors condition leads to the highest turnout at $32.9\%$.
This is almost five percentage points lower than in the experimental sample, accounting for the lower baseline propensity to vote.
Of particular note is that the effect of the Neighbors condition relative to control is lower than in the experiment.
The effect was $8.13$ percentage points in the experiment but only $6.31$ points in this subpopulation after adjustment.
A naive extrapolation using the treatment effect in the experiment would thus have been misleading.
The remaining treatment conditions follow a similar pattern
The turnout is lower after adjustment, and the effects relative to the control condition is lower than in the experiment.
Of note is the rank switch between the Civic Duty and Hawthorne conditions where the former had the lowest turnout among the postcard conditions in the experiment while the latter is lowest after the adjustment.


\section{Concluding remarks}

Matching is an important tool for empirical researchers, but the method is not always applicable.
Algorithms with guaranteed optimality properties have limited scope and require vast computational resources.
They are rarely useful when designs are complex or samples are large.
Investigators have therefore been forced to use alternative approaches to construct their matches, either by simplifying the problem or by using ad hoc procedures such as greedy matching.

We illustrate the issue with an extrapolation exercise of the treatment effects in a complex, large-scale experiment to an even larger population.
The results of the analysis indicates that the social pressure postcard would have lead to a noticeable increase in turnout if sent to all registered voters who voted in the 2004 partisan election.
The effect is, however, lower than if the results from the experiment were naively extrapolated.
Well-performing and computationally efficient methods of covariate adjustment are needed to do the extrapolation in a credible way.
The approach introduced in this paper provides a possible solution.

Generalized full matching is applicable to a wide range of studies.
Like its predecessor, the method admits good match quality without discarding large parts of the sample.
However, unlike conventional full matching, the method is not restricted to one particular design.
It can accommodate any number of treatment conditions and intricate compositional constraints over those conditions.
Studies with such designs have often solved several matching problems and merged the resulting matchings in a post-processing step.
For example, \citet{Silber2014} suggest matching the units in all treatment conditions to a common reference group, or a template in the authors' terminology.
Besides being tiring and error-prone, such approaches rarely maintain optimality even if the underlying methods are optimal with respect to each separate matching problem.
Generalized full matching allows the investigator to construct a single matching that directly corresponds to the desired design.
The algorithm used to construct these matchings, as implemented in the \texttt{quickmatch} package, uses computational resources efficiently.
This enables investigators to use the approach also in large studies where matching methods previously been infeasible.

We conclude these remarks by stressing that our algorithm should be used as a complement to existing matching methods.
There are settings where we would discourage its use.
Unlike existing approaches based on network flows \citep[see, e.g.,][]{Hansen2006}, the approach presented in this paper does not necessarily derive optimal solutions.
For this reason, best practice is still to use existing optimal algorithms when possible.
Furthermore, several refinements to the conventional full matching algorithm have been developed since its conception. \citet{Hansen2004} shows, for example, how to impose bounds on the ratio between the number of treated and control units within the matched groups.
This limits the weight variation of the matching and allows the investigator to directly control the bias-variance trade off related to how aggressive the adjustment is allowed to be.
When used with care, such control can greatly improve one's inferences.
While a similar effect can be achieved by adjusting the compositional constraints in a generalized full matching, it is a blunt solution without the same level of control as in \citeauthor{Hansen2004}'s formulation.

Network flow algorithms can also be adapted to construct matching with \emph{fine balance} \citep{Rosenbaum2007}.
The matched groups are here constructed so to ensure that the adjusted treatment groups have identical marginal distributions for a set of categorical covariates.
The current implementation of our algorithm cannot accommodate such global objectives.
\citet{Pimentel2015JASA} introduces a refinement of fine balancing where categorical covariates can be prioritized so that matches are constructed in a hierarchical fashion.
This ensures fine balance on covariates deemed more important before improving the balance in the rest of the covariate distribution.
The authors also show how large samples can be accommodated by thinning out the edges in the network flow problem.
Building on this idea, \citet{Yu2019} discuss a preprocessing procedure that allow investigators to use network flow algorithms with fine balancing constraints in moderately large studies with two treatment conditions.

It may be feasible to use algorithms with an exponential time complexity if the sample is sufficiently small.
One such example is \citet{Zubizarreta2012}, who provides a general framework for directly solving the integer programming problem that underlies the matching problem.
When feasible, this approach gives the investigator the greatest control of the matching problem which, when used with care, will allow for superior performance.
\citet{Bennett2018Building} show how the integer program in some instances can be relaxed to a more tractable linear program.
When used together with template matching to construct a small reference group, the approach can accommodate samples of several hundred thousand observations divided between more than two treatment conditions.

Investigators will, for good reasons, find these alternative matching approaches attractive in many situations.
They can, however, not be used with the complex compositional conditions and large samples accommodated by the method and algorithm introduced in this paper.
The extrapolation exercise from the voter mobilization experiment in Michigan is one such case.
We believe challenges of this type will become increasingly common as data sets grow in size, and we hope investigators will find the work presented here useful in such situations.

	\bibliographystyle{modapa}
	\bibliography{ref-nngmatching}

	\appendix

	\renewcommand\thetheorem{S\arabic{theorem}}
	\renewcommand\thesection{S\arabic{section}}
	\renewcommand\thetable{S\arabic{table}}
	\renewcommand\theequation{S\arabic{equation}}

	\setcounter{equation}{0}

	\section*{Proofs and additional results}

	\section{Brief overview of graph theory}

\begin{description}
	\item[Graph] A graph $G=(V,E)$ consists of a set of indices $V=\{a, b, \cdots\}$, called vertices, and a set of 2-element subsets of $V$, called edges.
	If an edge $\{i,j\}\in E$, we say that $i$ and $j$ are connected in the graph. In a directed graph (or digraph), the edges (which are then called arcs) are ordered sets $(i,j)$.
	In other words, $i$ can be connected to $j$ without the reverse being true in a digraph.

	\item[Weighted graph] A weighted graph assign a weight or cost to each edge or arcs.
	In our case, the weights are exclusively given by the distance of the connected vertices according to the distance metric used in the matching problem.

	\item[Adjacent] Vertices $i$ and $j$ are adjacent in $G$ if an edge (or an arc) connecting $i$ and $j$ exists in $E$.

	\item[Geodesic distance] The geodesic distance between $i$ and $j$ in $G$ is the number of edges or arcs (in our case, of any directionality) in the shortest path connecting $i$ and $j$ in $G$.

	\item[Subgraph] $G_1=(V_1,E_1)$ is a subgraph of $G_2=(V_2,E_2)$ if $V_1\subseteq V_2$ and $E_1\subseteq E_2$.
	In that case, we also say that $G_2$ is a supergraph of $G_1$. $G_1$ is a spanning subgraph of $G_2$ if $V_1=V_2$.

	\item[Complete graph] $G$ is complete if $\{i,j\}\in E$ for any two vertices $i,j\in V$.
	If $G$ is directed, both $(i,j)$ and $(j,i)$ must be in $E$.

	\item[Union] The union of $G_1=(V_1,E_1)$ and $G_2=(V_2,E_2)$ is $G_1\cup G_2=(V_1\cup V_2,E_1\cup E_2)$.

	\item[Graph difference] The difference between two graphs, $G_1=(V,E_1)$ and $G_2=(V,E_2)$, spanning the same set of vertices is $G_1 - G_2=(V,E_1\setminus E_2)$.

	\item[Independent set] A set of vertices $I \subseteq V$ is independent in $G$ if no two vertices in the set are adjacent:
	\begin{equation*}
	\forall \, i, j \in I,\; \{i,j\} \not\in E.
	\end{equation*}

	\item[Maximal independent set] An independent set of vertices $I$ in $G$ is maximal if for any additional vertex $i \in V$ the set $\{i\}\cup I$ is not independent:
	\begin{equation*}
	\forall \, i \in V \setminus I, \;\exists \, j \in I,\; \{i,j\} \in E.
	\end{equation*}

	\item[Cluster graph] The (directed) cluster graph induced by some partition of $V$ is the graph where arcs exists between any pair of units in the same component of the partition and no other arcs exist.

	\item[Adjancency matrix] The adjancency matrix $\mathbf{A}$ of a graph $G$ with $n$ vertices is a $n$-by-$n$ binary matrix where the entry $i,j$ is one if the edge $\{i,j\}$ or arc $(i,j)$ is in $E$, and otherwise zero.

\end{description}

\section{Proofs}

We here provide proofs for the propositions in the article.
The relevant propositions from the article are restated with their original numbering for reference.

\subsection{Optimality}

Recall the two objective functions:
\begin{eqnarray}
L^{Max}(\mathbf{M}) &=& \max_{\mathbf{m}\in\mathbf{M}}\max\{\dm(i,j) : i,j\in\mathbf{m}\},
\\
L^{Max}_{tc}(\mathbf{M}) &=& \max_{\mathbf{m}\in\mathbf{M}}\max\{\dm(i,j) : i,j\in\mathbf{m}\wedge W_i\neq W_j\}.
\end{eqnarray}

\begin{customlemma}{\ref{prop:closedN-admissible}}
	The closed neighborhood of each vertex in the $\mathcal{C}$-compatible nearest neighbor digraph satisfies the matching constraints $\mathcal{C}=(c_1, \cdots, c_k, t)$:
	\begin{equation}
	\forall i\in V,\, \forall j \in \{1, \cdots, k\},\, |\N[i] \cap \mathbf{w}_j| \geq c_j, \qquad \text{and} \qquad \forall i\in V,\, |\N[i]| \geq t.
	\end{equation}
\end{customlemma}

\begin{proof}
	This follows directly from the construction of $G_\mathcal{C}$.
	For each treatment-specific constraint, $c_1, \cdots, c_k$, the first step of the algorithm ensures that each vertex has that many arcs pointing to units assigned to the corresponding treatment condition.
	Similarly, if $t > c_1 + c_2 + \cdots + c_k$, the second step draws additional arcs so that each vertex has $t$ outward-pointing arcs in total.
\end{proof}

\begin{lemma} \label{prop:labeledvertex}
	Each vertex has at least one labeled vertex in its neighborhood in $G_\mathcal{C}$.
\end{lemma}

\begin{proof}
	By definition, all vertices in a closed neighborhood of a seed are labeled.
	That is, $\ell$ is a labeled vertex if and only if $\exists i\in \mathbf{S}, \ell \in \N[i]$.
	Suppose the lemma does not hold, i.e., that some vertex $i$ does not have a labeled vertex in its neighborhood:
	\begin{equation}
	\exists i, \forall \ell\in \N[i], \nexists j\in \mathbf{S}, \ell\in \N[j]. \tag{$*$}
	\end{equation}
	It follows directly that $i$ cannot be a seed as all vertices in its neighborhood would otherwise be labeled by definition.
	Note that ($*$) entails that $i$'s neighborhood does not have any overlap with any seed's neighborhood:
	\begin{equation}
	\forall j\in \mathbf{S}, \N[i]\cap \N[j] = \emptyset,
	\end{equation}
	However, this violates the maximality condition in the definition of a valid set of seeds (see the third step of the algorithm) and, subsequently, a vertex such as $i$ is not possible.
\end{proof}

\begin{lemma} \label{admissible}
	$\mathbf{M}_{alg}$ is an admissible generalized full matching with respect to the matching constraint $\mathcal{C}=(c_1, \cdots, c_k, t)$:
	\begin{equation}
	\mathbf{M}_{alg}\in \mathcal{M}_\mathcal{C}.
	\end{equation}
\end{lemma}

\begin{proof}
	We must show that $\mathbf{M}_{alg}$ satisfies the four conditions of an admissible generalized full matching in Definition \ref{def:agfm}.

	Step 5 of the algorithm ensures that $\mathbf{M}_{alg}$ is spanning.
	At this step, any vertex that lacks a label will be assigned the same label as one of the labeled vertices in its neighborhood.
	Lemma \ref{prop:labeledvertex} ensures that at least one labeled vertex exists in the neighborhoods of the unassigned vertices.

	No vertex is assigned more than one label; $\mathbf{M}_{alg}$ is disjoint.
	Vertices are only assigned labels in either Step 4 or 5, but never in both.
	Step 3 ensures that the neighborhoods of the seeds are non-overlapping.
	Thus vertices will be assigned at most one label in Step 4.
	In Step 5, vertices are explicitly assigned only one label even if several labels could be represented in a vertex's neighborhood.

	The remaining two conditions in Definition \ref{def:agfm} are ensured by Lemma \ref{prop:closedN-admissible}.
	Step 4 of the algorithm ensures that each matched group is a superset of a seed's neighborhood.
	From Lemma \ref{prop:closedN-admissible}, we have that this neighborhood will satisfy the matching constraints.
\end{proof}

\begin{lemma} \label{boundedby4}
	If the arc weights in the $\mathcal{C}$-compatible nearest neighbor digraph are bounded by some $\lambda$, the maximum within-group distance in $\mathbf{M}_{alg}$ is bounded by $4\lambda$:
	\begin{equation}
	\forall (i,j)\in E_\mathcal{C},\, \dm(i,j) \leq \lambda \implies \max_{\mathbf{m}\in\mathbf{M}_{alg}}\max\{\dm(i,j) : i,j\in\mathbf{m}\} \leq 4\lambda.
	\end{equation}
\end{lemma}

\begin{proof}
	First, consider the distance from any vertex $i$ to the seed in its matched group, denoted $j$.
	If $i$ is a seed, we have $i=j$ as each matched group contains exactly one seed by construction.
	By the self-similarity property of distance metrics, the distance is zero: $\dm(i,j)=0$.
	If $i$ is a labeled vertex (i.e., assigned a label in Step 4 of the algorithm), we have $(j,i)\in E_\mathcal{C}$ by definition of labeled vertices.
	By assumption, $\dm(j,i)$ is bounded by $\lambda$.
	Due to the symmetry property of distance metrics, this also bounds $\dm(i,j)$.
	If $i$ is not a labeled vertex (i.e., assigned a label in Step 5), it will be adjacent in $G_\mathcal{C}$ to a labeled vertex, $\ell$, in its matched group as implied by Lemma \ref{prop:labeledvertex}.
	We have $(i, \ell)\in E_\mathcal{C}$ so the distance between $i$ and $\ell$ is bounded by $\lambda$.
	As $\ell$ is labeled, we have $(j, \ell)\in E_\mathcal{C}$ which implies that $\dm(j, \ell) \leq \lambda$.
	From the triangle inequality property of metrics, we have that the distance between $i$ and the seed, $j$, is at most $2\lambda$.

	Now consider any two vertices assigned to the same matched group.
	We have shown that the distance from each of these vertices to their (common) seed is at most $2\lambda$.
	By applying the triangle inequality once more, we bound the distance between the two non-seed vertices by $4\lambda$.
\end{proof}

\begin{lemma} \label{ccissmallest}
	The $\mathcal{C}$-compatible nearest neighbor digraph has the smallest maximum arc weight among all digraphs compatible with $\mathcal{C}$, i.e., all graph in which each vertex's closed neighborhood contains $c_1, c_2, \cdots, c_k$ vertices of each treatment condition and $t$ vertices in total.
\end{lemma}

\begin{proof}
	The definition of the directed neighborhoods is asymmetric in the sense that if $i$ is in $j$'s neighborhood, $j$ is not necessarily in $i$'s neighborhood.
	Thus, whether a vertex's neighborhood satisfies the constraints is independent of whether other vertices' neighborhoods do so.
	As a consequence, to minimize the maximum arc weight, we can simply minimize the maximum arc weight in each neighborhood separately.
	To minimize the arc weights in a single neighborhood, we draw arcs to the vertices closest to the vertex so that the matching constraints are fulfilled.
	This is exactly the procedure the algorithm follows.
\end{proof}

\begin{customlemma}{\ref{prop:NNGbound}}
	The distance between any two vertices connected by an arc in the $\mathcal{C}$-compatible nearest neighbor digraph, $G_\mathcal{C} = (V, E_\mathcal{C})$, is less or equal to the maximum within-group distance in an optimal matching:
	\begin{equation}
	\forall (i,j)\in E_\mathcal{C},\; \dm(i,j) \leq \min_{\mathbf{M}\in\mathcal{M}_\mathcal{C}} L^{Max}(\mathbf{M}).
	\end{equation}
\end{customlemma}

\begin{proof}
	Let $w^*$ be the maximum within-group distance in an optimal matching and let $w_\mathcal{C}^+$ be the maximum weight of an arc in $G_\mathcal{C}$:
	\begin{align*}
	w^* &=\min_{\mathbf{M}\in\mathcal{M}_\mathcal{C}} L^{Max}(\mathbf{M}),
	\\
	w_\mathcal{C}^+ &= \max\{\dm(i,j): (i,j)\in E_\mathcal{C}\}.
	\end{align*}
	Furthermore, let $B_\mathcal{C} = (\mathbf{U}, E_{\mathcal{C}}^b)$ be the digraph that contains arcs between all units at a distance strictly closer than $w_\mathcal{C}^+$:
	\begin{equation}
	E_\mathcal{C}^b = \{(i,j): \dm(i,j) < w_\mathcal{C}^+\}.
	\end{equation}
	$B_\mathcal{C}$ must contain a vertex whose neighborhood does not satisfy the size constraints.
	If no such vertex exists, a digraph compatible with $\mathcal{C}$ with a smaller maximum arc weight than in $G_\mathcal{C}$ exists as a subgraph of $B_\mathcal{C}$.
	This contradicts Lemma \ref{ccissmallest}.

	Let $B_{op} = (\mathbf{U}, E_{op}^b)$ be the digraph that contains arcs between all units at a distance weakly closer than $w^*$:
	\begin{equation}
	E_{op}^b = \{(i,j): \dm(i,j) \leq w^*\}.
	\end{equation}
	By construction, $B_{op}$ is a supergraph of the cluster graph induced by the optimal matching.
	That is, arcs are drawn in $B_{op}$ between all units assigned to the same matched group in the optimal matching.
	As the optimal matching is admissible, each vertex's neighborhood in $B_{op}$ is compatible with $\mathcal{C}$.

	Suppose the lemma does not hold: $w_\mathcal{C}^+ > w^*$.
	It follows that $E_{op}^b \subset E_\mathcal{C}^b$.
	As at least one vertex's neighborhood does not satisfy the size constraint in $B_\mathcal{C}$, that must be the case in $B_{op}$.
	This, however, implies that the optimal matching is not admissible which, in turn, contradicts optimality.
\end{proof}

\begin{customtheorem}{\ref{prop:approx-opt}}
	$\mathbf{M}_{alg}$ is a 4-approximate generalized full matching with respect to the matching constraint $\mathcal{C}=(c_1, \cdots, c_k, t)$ and matching objective $L^{Max}$:
	\begin{equation}
	\mathbf{M}_{alg}\in \mathcal{M}_\mathcal{C}, \qquad \text{and} \qquad L^{Max}(\mathbf{M}_{alg}) \leq \min_{\mathbf{M}\in\mathcal{M}_\mathcal{C}} 4 L^{Max}(\mathbf{M}).
	\end{equation}
\end{customtheorem}

\begin{proof}
	Admissibility follows from Lemma \ref{admissible}. Approximate optimality follows from Lemma \ref{boundedby4} and \ref{prop:NNGbound}.
\end{proof}

\begin{lemma} \label{CClessthanM2}
	When all treatment-specific constraints are less or equal to one and the overall size constraint is the sum of the treatment-specific constraints, the distance between any two vertices connected by an arc in $G_\mathcal{C} = (V, E_\mathcal{C})$ is less or equal to the maximum within-group distance in an optimal matching with $L^{Max}_{tc}$ as objective:
	\begin{equation}
	c_1, c_2, \cdots, c_k \leq 1 \wedge t = {\textstyle\sum_{x=1}^k c_x} \Rightarrow \forall (i,j)\in E_\mathcal{C}, \dm(i,j) \leq \min_{\mathbf{M}\in\mathcal{M}_\mathcal{C}} L^{Max}_{tc}(\mathbf{M}).
	\end{equation}
\end{lemma}

\begin{proof}
	Let $w_{s}^+$ be the maximum weight of an arc connecting two units with the same treatment conditions in $G_\mathcal{C}$, let $w_{d}^+$ the maximum arc weight between units with different conditions:
	\begin{align*}
	w_{s}^+ &= \max\{\dm(i,j): (i,j)\in E_\mathcal{C} \wedge W_i = W_j\},
	\\
	w_{d}^+ &= \max\{\dm(i,j): (i,j)\in E_\mathcal{C} \wedge W_i \neq W_j\}.
	\end{align*}
	Note that:
	\begin{equation*}
	\max\{\dm(i,j): (i,j)\in E_\mathcal{C}\} = \max\{w_{s}^+, w_{d}^+\}.
	\end{equation*}

	Consider $w_{s}^+$.
	Since $c_1, c_2, \cdots, c_k \leq 1$ and $t = \sum_{x=1}^k c_x$, each unit will have at most one arc pointing to a unit with the same treatment condition as its own:
	\begin{equation}
	\forall i,\; |\{(i,j): (i,j)\in E_\mathcal{C} \wedge W_i = W_j\}|=c_{W_i} \leq 1.
	\end{equation}
	From the self-similarity and non-negativity properties of distance metrics, we have:
	\begin{equation}
	\forall i,j,\; 0=\dm(i,i) \leq \dm(i,j).
	\end{equation}
	By construction of $G_\mathcal{C}$, all arcs in the set will be self-loops and, thus, at distance zero:
	\begin{equation}
	w_{s}^+ = \max\{\dm(i,i): (i,i)\in E_\mathcal{C}\} = 0.
	\end{equation}
	From non-negativity, it follows that:
	\begin{equation*}
	\max\{\dm(i,j): (i,j)\in E_\mathcal{C}\} = \max\{0, w_{d}^+\} = w_{d}^+.
	\end{equation*}

	Let $w^*$ be the maximum within-group distance between units assigned to different treatment conditions when $L^{Max}_{tc}$ is used as objective:
	\begin{equation*}
	w^* = \min_{\mathbf{M}\in\mathcal{M}_\mathcal{C}} L^{Max}_{tc}(\mathbf{M}).
	\end{equation*}
	Let $B_{d} = (\mathbf{U}, E_{d}^b)$ be the digraph that contains all arcs between units that either are strictly closer than $w_{d}^+$ or have the same treatment condition:
	\begin{equation}
	E_{d}^b = \{(i,j): \dm(i,j) < w_{d}^+ \vee W_i = W_j\}.
	\end{equation}
	Following the same logic as in the proof of Lemma \ref{prop:NNGbound}, $B_{d}$ must contain a vertex whose neighborhood is not compatible with $\mathcal{C}$.

	Let $B_{op} = (\mathbf{U}, E_{op}^b)$ be the digraph that contains all arcs between units that either are weakly closer than $w^*$ or have the same treatment condition:
	\begin{equation}
	E_{op}^b = \{(i,j): \dm(i,j) \leq w^* \vee W_i = W_j\}.
	\end{equation}
	By construction, $B_{op}$ is a supergraph of the cluster graph induced by the optimal matching.
	That is, arcs are drawn in $B_{op}$ between all units assigned to the same matched group in the optimal matching.
	As the optimal matching is admissible, each vertex's neighborhood in $B_{op}$ is compatible with $\mathcal{C}$.

	Assume $w_{d}^+ > w^*$. It follows that $E_{op}^b \subset E_{d}^b$.
	As at least one vertex's neighborhood does not satisfy the size constraint in $B_{d}$, that must be the case in $B_{op}$.
	This, however, implies that the optimal matching is not admissible which, in turn, contradicts optimality.
	We conclude that $w_{d}^+ \leq w^*$.
\end{proof}

\begin{customtheorem}{\ref{prop:approx-opt-trad}}
	$\mathbf{M}_{alg}$ is a 4-approximate conventional full matching with respect to the matching constraint $\mathcal{C}=(1, \cdots, 1, k)$ and matching objective $L^{Max}_{tc}$:
	\begin{equation}
	\mathbf{M}_{alg}\in \mathcal{M}_\mathcal{C}, \qquad \text{and} \qquad L^{Max}_{tc}(\mathbf{M}_{alg}) \leq \min_{\mathbf{M}\in\mathcal{M}_\mathcal{C}} 4 L^{Max}_{tc}(\mathbf{M}).
	\end{equation}
\end{customtheorem}

\begin{proof}
	Admissibility follows from Lemma \ref{admissible}.
	Note that all distances considered by $L^{Max}_{tc}$ are considered by $L^{Max}$ as well.
	As a result, the latter acts as a bound for the former:
	\begin{equation}
	\forall\mathbf{M}\in\mathcal{M}_{\mathcal{C}}, L^{Max}_{tc}(\mathbf{M}) \leq L^{Max}(\mathbf{M}).
	\end{equation}
	Approximate optimality follows from Lemma \ref{boundedby4} and \ref{CClessthanM2}:
	\begin{equation*}
		L^{Max}_{tc}(\mathbf{M}_{alg}) \leq L^{Max}(\mathbf{M}_{alg}) \leq4\min_{\mathbf{M}\in\mathcal{M}_\mathcal{C}} L^{Max}_{tc}(\mathbf{M}).
	\end{equation*}
\end{proof}

\subsection{Complexity}

\begin{lemma} \label{ccinpolytime}
	A $\mathcal{C}$-compatible nearest neighbor digraph can be constructed in polynomial time using linear memory.
\end{lemma}

\begin{proof}
	In the first step of the algorithm, we construct $G_w$ as the union of $\NN(c_x, G(\mathbf{U}\rightarrow \mathbf{w}_x))$ for each treatment condition $x$.
	The operands of this union can be constructed using nearest neighbor searches for each treatment condition.
	With a naive implementation, such searches can be done sequentially for each $i\in\mathbf{U}$ by sorting the set $\{\dm(i,j): j\in\mathbf{w}_x\}$ and drawing an arc from $i$ to the first $c_x$ elements in the sorted set.
	When using standard sorting algorithms, this has a time complexity of $O(n|\mathbf{w}_x| \log |\mathbf{w}_x|)$ and a space complexity of $O(c_x n)$ \citep{Knuth1998}.
	Note that $|\mathbf{w}_x| < n$ for all treatments, so the search requires $O(n^2 \log n)$ time.
	The union can be performed in linear time in the total number of arcs, $O[(c_1 + c_2 + \cdots + c_k) n]$.
	As each $\NN(c_x, G(\mathbf{U}\rightarrow \mathbf{w}_x))$ can be derived sequentially and the size constraints are fixed, the $G_w$ digraph can be constructed in $O(n^2 \log n)$ time.

	In the second step, $G_r$ can be constructed in a similar fashion.
	For each $i\in\mathbf{U}$, sort the set $\{\dm(i,j): j\in\mathbf{U} \wedge (i,j)\not\in E_{w}\}$ and draw an arc from $i$ to the first $r = t - c_1 - \cdots - c_k$ elements in that set.
	Like above, this has a complexity of $O(n^2 \log n)$.
	Finally, the union between $G_{w}$ and $G_{r}$ can be constructed in linear time in the total number of arcs.
	As the number of arcs per vertex is fixed at $t$, the union is completed in $O(n)$ time.
	The steps are sequential so the total complexity of both Step 1 and 2 is $O(n^2 \log n)$.
\end{proof}

\begin{remark} \label{ccinqlineartime}
	For most common metrics, standard sorting algorithms are inefficient.
	Storing the data points in a structure made for the purpose, such as a kd- or bd-tree, typically leads to large improvements \citep{Friedman1977}.
	Each $\NN(c_x, G(\mathbf{U}\rightarrow \mathbf{w}_x))$ can then be constructed in $O(n \log n)$ average time, without changing the memory complexity.
	However, this approach typically requires a preprocessing step to build the search tree.
	In the proof of Lemma \ref{ccinpolytime}, the search set is unique for each vertex when $G_{r}$ is constructed.
	We can, therefore, not use these specialized algorithms if we construct $G_{r}$ in the way suggested there.
	However, the construction can easily be transformed into a problem with a fixed search set. Note that:
	\begin{equation}
	\NN(r, G(\mathbf{U} \rightarrow \mathbf{U}) - G_w) = \NN(r, \NN(t, G(\mathbf{U} \rightarrow \mathbf{U})) - G_{w}).
	\end{equation}
	That is, finding the $r$ nearest neighbors not already connected in $G_{w}$ is the same as finding the $r$ nearest neighbors not already connected in $G_{w}$ among the $t$ nearest neighbors in the complete graph.
	The first nearest neighbor search, $\NN(t, G(\mathbf{U} \rightarrow \mathbf{U}))$, has a fixed search set and can thus be completed in $O(n \log n)$.
	The second nearest neighbor search involves sorting at most $t$ elements for each vertex, which is done in constant time as $t$ is fixed.
\end{remark}

\begin{customtheorem}{\ref{prop:complexity}}
	In the worst case, the generalized full matching algorithm terminates in polynomial time using linear memory.
\end{customtheorem}

\begin{proof}
	The algorithm runs sequentially.
	The first and second step can be completed in $O(n^2 \log n)$ worst case time as shown in Lemma \ref{ccinpolytime}, or, in many cases, in $O(n \log n)$ average time as discussed in Remark \ref{ccinqlineartime}.

	Step 3 and 4 can be done by sequentially labeling seeds and their neighbors as they are selected.
	Any vertex whose neighborhood does not contain any labeled vertices can be a valid seed, and any vertex that is adjacent to labeled vertices can never become a seed.
	Thus, traversing the vertices in any order and greedily selecting units as seed will yield a valid set of seeds.
	As the size of each seed's neighborhood is fixed at $t$, this step is completed in $O(n)$ time.

	Finally, assigning labels to unlabeled vertices in the last step can be done by traversing over their neighborhoods.
	Thus, Step 5 also requires $O(n)$ time to complete.
\end{proof}

\clearpage

\section{Additional simulation results}

The following tables present additional results from the simulation study.
Sections~\ref{sec:additional-distances} and~\ref{sec:additional-group} provide results about aggregated distances for the algorithms and the structure of the matched groups they produce.
Section~\ref{sec:additional-main} gives complete results for the measures presented in the paper.
These tables also include the results for 1:2-matching without replacement.

\subsection{Distances}\label{sec:additional-distances}

We investigate five different functions aggregating within-group distances:
\allowdisplaybreaks
\begin{align*}
L^{Max}(\mathbf{M}) &= \max_{\mathbf{m}\in\mathbf{M}}\max\{\dm(i,j) : i,j\in\mathbf{m}\},
\\[1ex]
L^{Max}_{tc}(\mathbf{M}) &= \max_{\mathbf{m}\in\mathbf{M}}\max\{\dm(i,j) : i,j\in\mathbf{m}\wedge W_i\neq W_j\},
\\[1ex]
L^{Mean}(\mathbf{M}) &= \sum_{\mathbf{m}\in\mathbf{M}}\frac{|\mathbf{w}_1\cap\mathbf{m}|}{|\mathbf{w}_1|}\mean\{\dm(i,j) :i,j\in\mathbf{m} \wedge i\neq j\},
\\[1ex]
L^{Mean}_{tc}(\mathbf{M}) &= \sum_{\mathbf{m}\in\mathbf{M}}\frac{|\mathbf{w}_1\cap\mathbf{m}|}{|\mathbf{w}_1|}\mean\{\dm(i,j) :i,j\in\mathbf{m} \wedge W_i\neq W_j\},
\\[1ex]
L^{Sum}_{tc}(\mathbf{M}) &= \sum_{\mathbf{m}\in\mathbf{M}}\sum\{\dm(i,j) :i,j\in\mathbf{m} \wedge W_i\neq W_j\}.
\end{align*}
$L^{Max}$ is the maximum within-group distance between any two units, and $L^{Max}_{tc}$ is the maximum distance between treated and control units.
They are the objectives discussed in Section \ref{sec:objective} and are the ones used by the \texttt{quickmatch} package.
$L^{Mean}_{tc}$ is the average within-group distance between treated and control units weighted by the number treated units in the groups.
It is the objective function discussed by \citet{Rosenbaum1991} when he introduced full matching.
As \citeauthor{Rosenbaum1991} notes, this objective is neutral in the sense that the size of the matched groups matters only insofar as it affects the within-group distances.
To contrast with $L^{Max}$, we include $L^{Mean}$: a version of the mean distance objective that also considers within-group distances between units assigned to the same treatment condition.

Finally, $L^{Sum}_{tc}$ is the sum of within-group distances between treated and control units.
With the terminology of \citet{Rosenbaum1991}, this function favors small subclasses and is, thus, not neutral.
As a consequence, if we were to use $L^{Sum}_{tc}$ as our objective, we would accept matchings with worse balance if the matched groups were sufficiently smaller.
When the matching structure is fixed (as with 1:1- and 1:k-matching without replacement), $L^{Sum}_{tc}$ is proportional to $L^{Mean}_{tc}$ and, thus, identical for practical purposes.
Both the \texttt{optmatch} and \texttt{Matching} packages use the sum as their objective.

Table \ref{tab:distances} presents the distance measures for the different methods.
As distances have no natural scale, we normalize the results by the results of conventional full matching in smaller sample.
We see that 1:1-matching with replacement greatly outperforms the other methods, especially on $L^{Sum}_{tc}$ which is the objective function it uses.
The implementations of both conventional and generalized full matching perform largely the same, with a slight advantage to \texttt{optmatch} on the $L^{Sum}_{tc}$ measure.
All versions of matching without replacement performs considerably worse than the other methods, in particular on the measures they do not use as their objective.
Predictively, the optimal implementations produce shorter distances than the greedy versions, but the differences are small.

Comparisons in aggregated distances between methods that impose different matching constraints can be awkward because the methods solve different types of matching problems.
For example, 1:2-matching will necessarily lead to larger distances than 1:1-matching, but the former can be preferable if, for example, we are interested in \textsc{att} and control units vastly outnumbers treated units.
Comparisons between methods using the same matching constraints should, however, be informative.

\begin{table}[ht]
	\centering
	\caption{Aggregated distances for matching methods with samples of 1,000 and 10,000 units.}\label{tab:distances}
	\resizebox{\textwidth}{!}{%
		\begin{tabular}{l c rrrrr c rrrrr}
			                & &                                                                           \multicolumn{5}{c}{\underline{1,000 units}}                                                                           & &                                                                           \multicolumn{5}{c}{\underline{10,000 units}}                                                                           \\
                & &        \multicolumn{1}{r}{$L^{Max}$} &  \multicolumn{1}{r}{$L^{Max}_{tc}$} &      \multicolumn{1}{r}{$L^{Mean}$} & \multicolumn{1}{r}{$L^{Mean}_{tc}$} &  \multicolumn{1}{r}{$L^{Sum}_{tc}$} & &        \multicolumn{1}{r}{$L^{Max}$} &  \multicolumn{1}{r}{$L^{Max}_{tc}$} &      \multicolumn{1}{r}{$L^{Mean}$} & \multicolumn{1}{r}{$L^{Mean}_{tc}$} &  \multicolumn{1}{r}{$L^{Sum}_{tc}$}  \\ \cline{3-7} \cline{9-13}
Greedy 1:1      & &                                 1.87 &                                 2.67 &                                 1.41 &                                 1.50 &                                 0.43 & &                                 2.20 &                                 3.14 &                                 0.89 &                                 0.95 &                                 2.69  \\
Optimal 1:1     & &                                 1.29 &                                 1.85 &                                 1.20 &                                 1.27 &                                 0.36 & &                                 1.87 &                                 2.68 &                                 0.80 &                                 0.85 &                                 2.41  \\
Replacement 1:1 & &                                 0.45 &                                 0.51 &                                 0.65 &                                 0.66 &                                 0.19 & &                                 0.19 &                                 0.20 &                                 0.20 &                                 0.21 &                                 0.59  \\
Greedy 1:2      & &                                 3.66 &                                 5.23 &                                 3.21 &                                 4.31 &                                 2.46 & &                                 3.99 &                                 5.71 &                                 2.51 &                                 3.69 &                                20.97  \\
Optimal 1:2     & &                                 3.27 &                                 4.68 &                                 3.17 &                                 3.79 &                                 2.17 & &                                 3.93 &                                 5.62 &                                 2.96 &                                 3.50 &                                19.87  \\
Full matching   & &                                 1.00 &                                 1.00 &                                 1.00 &                                 1.00 &                                 1.00 & &                                 0.39 &                                 0.38 &                                 0.31 &                                 0.31 &                                 3.10  \\
GFM             & &                                 1.00 &                                 1.00 &                                 0.99 &                                 0.98 &                                 1.05 & &                                 0.39 &                                 0.38 &                                 0.31 &                                 0.30 &                                 3.25  \\
Refined GFM     & &                                 0.95 &                                 1.25 &                                 0.98 &                                 1.10 &                                 1.19 & &                                 0.37 &                                 0.49 &                                 0.31 &                                 0.34 &                                 3.70  \\
 \hline
			\multicolumn{13}{p{0.9\textwidth}}{\scriptsize\emph{Notes:} The measures are normalized by the result for conventional full matching in the sample with 1,000 units. Results are based on 10,000 simulation rounds. Simulation errors are negligible.} \\
		\end{tabular}
	}
\end{table}

\subsection{Group structure}\label{sec:additional-group}

Table \ref{tab:group-structure} presents measures of the group structure for the different matching methods.
The first measure is the average size of the matched groups. 1:1- and 1:2-matching without replacement have a fixed group size of either two or three units.
The group size for matching with replacement depends on the sparseness of the control units. Overlap is reasonably good with the current data generating process, and the average group size increases with only 20\% compared to matching without replacement.
The full matching methods lead to larger groups since they do not discard units.
Given the unconditional propensity score of 26.5\%, the expected minimum group size among matchings that do not discard units is 3.77 units, which is close to what the methods produce.
The groups are slightly smaller with conventional full matching.
This is likely a result of both that implementation's optimality and its use of a non-neutral objective function (i.e., $L^{Sum}_{tc}$).
In the second column, we present the standard deviation of the group sizes. We see that the full matching methods have considerably higher variation.
This is a result of their ability to adapt the matching to the distribution of units in the covariate space.

Next, we investigate the share of the sample that is discarded.
For a given level of balance, we want to drop as few units as possible.
Predictably, 1:1-matching leads to that a sizable portion of the sample are left unassigned.
This is especially the case when we match with replacement.
Fewer units are discarded with 1:2-matching, and by construction, no units are discarded with the full matching methods.

The fourth column reports the standard deviation of the weights implicitly used for the adjustment in the estimator.
Weight variation is necessary to balance an unbalanced sample.
However, for a given level of balance, we want the weights to be as uniform as possible.
Since we are estimating \textsc{att}, the implied weights for treated units are fixed at $|\mathbf{w}_1|^{-1}$ for all methods.
Weights for control do, however, vary.
The implied weight for control unit $i$ assigned to matched group $\mathbf{m}$ is:
\begin{equation}
\text{wgh}_i = \frac{|\mathbf{w}_1\cap\mathbf{m}|}{|\mathbf{w}_1| \times |\mathbf{w}_0\cap\mathbf{m}|},
\end{equation}
and zero if not assigned to a group.
Examining the results, we see that the amount of variation is correlated with how well the methods are able to minimize distances.
For example, 1:1-matching with replacement produces the shortest distances, but as a result, also the most weight variation.
The choice of method depends on how one resolves the trade-off between weight variation and balance, which, in turn, depends on how strongly the covariates are correlated with the outcome and treatment assignment.
For this reason, the best choice of matching method will differ depending on the data generating process.
It appears, however, that all full matching methods lead to matchings with substantially smaller distances than 1:1-matching without replacement with only slightly higher weight variation (i.e., close to a Pareto improvement).
Similarly, but less pronounced, the \texttt{optmatch} package dominates the \texttt{quickmatch} package; the former produces about the same distances but with less weight variation.

\begin{table}[ht]
\centering
\caption{Group composition for matching methods with samples of 1,000 and 10,000 units.}\label{tab:group-structure}
\resizebox{\textwidth}{!}{%
	\begin{tabular}{l c rrrr c rrrr}
		                & &                                                                        \multicolumn{4}{c}{\underline{1,000 units}}                                                                       & &                                                                       \multicolumn{4}{c}{\underline{10,000 units}}                                                                        \\
                & &                     \multicolumn{1}{r}{Size} & \multicolumn{1}{r}{$\sigma(\text{Size})$} &                \multicolumn{1}{r}{\% drop} &  \multicolumn{1}{r}{$\sigma(\text{wgh})$} & &                     \multicolumn{1}{r}{Size} & \multicolumn{1}{r}{$\sigma(\text{Size})$} &                \multicolumn{1}{r}{\% drop} &  \multicolumn{1}{r}{$\sigma(\text{wgh})$}  \\ \cline{3-6} \cline{8-11}
Greedy 1:1      & &                                         2.00 &                                         0.00 &                                        46.96 &                                         1.81 & &                                         2.00 &                                         0.00 &                                        47.03 &                                         1.81  \\
Optimal 1:1     & &                                         2.00 &                                         0.00 &                                        46.96 &                                         1.81 & &                                         2.00 &                                         0.00 &                                        47.03 &                                         1.81  \\
Replacement 1:1 & &                                         2.41 &                                         0.86 &                                        54.70 &                                         2.85 & &                                         2.41 &                                         0.87 &                                        54.73 &                                         2.85  \\
Greedy 1:2      & &                                         3.00 &                                         0.00 &                                        20.44 &                                         0.84 & &                                         3.00 &                                         0.00 &                                        20.54 &                                         0.85  \\
Optimal 1:2     & &                                         3.00 &                                         0.00 &                                        20.44 &                                         0.84 & &                                         3.00 &                                         0.00 &                                        20.54 &                                         0.85  \\
Full matching   & &                                         4.24 &                                         3.50 &                                         0.00 &                                         1.93 & &                                         4.24 &                                         3.51 &                                         0.00 &                                         1.97  \\
GFM             & &                                         4.74 &                                         3.54 &                                         0.00 &                                         2.13 & &                                         4.73 &                                         3.55 &                                         0.00 &                                         2.15  \\
Refined GFM     & &                                         4.55 &                                         3.26 &                                         0.00 &                                         2.04 & &                                         4.54 &                                         3.30 &                                         0.00 &                                         2.07  \\
 \hline
		\multicolumn{11}{p{0.9\textwidth}}{\scriptsize\emph{Notes:} The columns report the average group size, the standard deviation of the size, share of units not assigned to a group and the standard deviation in the weights of the control units implied by the matchings. Results are based on 10,000 simulation rounds. Simulation errors are negligible.} \\
	\end{tabular}
}
\end{table}

\clearpage

\subsection{Complexity and matching quality}\label{sec:additional-main}
.
\begin{table}[ht]
\centering
\caption{Covariate balance for matching methods with samples of 1,000 and 10,000 units.}
\resizebox{\textwidth}{!}{%
	\begin{tabular}{l c rrrrr c rrrrr}
		                & &                                                          \multicolumn{5}{c}{\underline{1,000 units}}                                                         & &                                                         \multicolumn{5}{c}{\underline{10,000 units}}                                                          \\
                & &     \multicolumn{1}{r}{$X_1$} &    \multicolumn{1}{r}{$X_2$} &  \multicolumn{1}{r}{$X_1^2$} &  \multicolumn{1}{r}{$X_2^2$} & \multicolumn{1}{r}{$X_1X_2$} & &     \multicolumn{1}{r}{$X_1$} &    \multicolumn{1}{r}{$X_2$} &  \multicolumn{1}{r}{$X_1^2$} &  \multicolumn{1}{r}{$X_2^2$} & \multicolumn{1}{r}{$X_1X_2$}  \\ \cline{3-7} \cline{9-13}
Unadjusted      & &                         52.48 &                         52.73 &                         10.72 &                         10.91 &                         13.11 & &                        52.528 &                        52.735 &                        10.707 &                        10.874 &                        12.588  \\
Greedy 1:1      & &                          5.93 &                          5.94 &                          7.21 &                          7.31 &                         13.87 & &                         5.270 &                         5.286 &                         6.647 &                         6.756 &                        13.332  \\
Optimal 1:1     & &                          5.94 &                          5.95 &                          7.08 &                          7.19 &                         14.09 & &                         5.260 &                         5.279 &                         6.594 &                         6.704 &                        13.396  \\
Replacement 1:1 & &                          0.44 &                          0.44 &                          0.76 &                          0.76 &                          0.80 & &                         0.043 &                         0.043 &                         0.077 &                         0.079 &                         0.077  \\
Greedy 1:2      & &                         26.23 &                         26.39 &                         15.48 &                         15.76 &                         35.59 & &                        25.662 &                        25.741 &                        15.410 &                        15.667 &                        36.914  \\
Optimal 1:2     & &                         26.18 &                         26.34 &                         15.22 &                         15.48 &                         36.11 & &                        25.668 &                        25.759 &                        15.495 &                        15.749 &                        36.745  \\
Full matching   & &                          1.00 &                          1.00 &                          1.00 &                          1.00 &                          1.00 & &                         0.105 &                         0.105 &                         0.106 &                         0.108 &                         0.096  \\
GFM             & &                          0.74 &                          0.75 &                          0.77 &                          0.77 &                          0.80 & &                         0.074 &                         0.075 &                         0.075 &                         0.077 &                         0.073  \\
Refined GFM     & &                          1.04 &                          1.05 &                          0.99 &                          0.99 &                          1.08 & &                         0.108 &                         0.108 &                         0.103 &                         0.104 &                         0.105  \\
 \hline
		\multicolumn{13}{p{0.95\textwidth}}{\scriptsize\emph{Notes:} The measures are normalized by the result for conventional full matching in the sample with 1,000 units.} \\
	\end{tabular}
}
\end{table}

\begin{table}[ht]
\centering
\caption{Estimator performance for matching methods with samples of 1,000 and 10,000 units.}
\resizebox{0.9\textwidth}{!}{%
	\begin{tabular}{l c rrrr c rrrr}
		                & &                                                                                                                      \multicolumn{4}{c}{\underline{1,000 units}}                                                                                                                     & &                                                                                                                     \multicolumn{4}{c}{\underline{10,000 units}}                                                                                                                      \\
                & &                                            \multicolumn{1}{r}{Bias} &                                   \multicolumn{1}{r}{\textsc{se}} &                                 \multicolumn{1}{r}{\textsc{rmse}} & \multicolumn{1}{r}{$\frac{\text{Bias}}{\text{\textsc{rmse}}}$} & &                                            \multicolumn{1}{r}{Bias} &                                   \multicolumn{1}{r}{\textsc{se}} &                                 \multicolumn{1}{r}{\textsc{rmse}} & \multicolumn{1}{r}{$\frac{\text{Bias}}{\text{\textsc{rmse}}}$}  \\ \cline{3-6} \cline{8-11}
Unadjusted      & &                                                               83.34 &                                                                1.47 &                                                               12.70 &                                                               0.993 & &                                                              83.390 &                                                                0.47 &                                                               12.64 &                                                               0.999  \\
Greedy 1:1      & &                                                                4.86 &                                                                1.04 &                                                                1.26 &                                                               0.583 & &                                                               4.118 &                                                                0.33 &                                                                0.71 &                                                               0.884  \\
Optimal 1:1     & &                                                                4.96 &                                                                1.04 &                                                                1.27 &                                                               0.590 & &                                                               4.153 &                                                                0.33 &                                                                0.71 &                                                               0.885  \\
Replacement 1:1 & &                                                                0.11 &                                                                1.17 &                                                                1.16 &                                                               0.015 & &                                                               0.024 &                                                                0.38 &                                                                0.37 &                                                               0.010  \\
Greedy 1:2      & &                                                               33.97 &                                                                1.61 &                                                                5.38 &                                                               0.956 & &                                                              32.980 &                                                                0.52 &                                                                5.02 &                                                               0.995  \\
Optimal 1:2     & &                                                               34.08 &                                                                1.59 &                                                                5.39 &                                                               0.957 & &                                                              32.939 &                                                                0.51 &                                                                5.01 &                                                               0.995  \\
Full matching   & &                                                                1.00 &                                                                1.00 &                                                                1.00 &                                                               0.151 & &                                                               0.077 &                                                                0.32 &                                                                0.32 &                                                               0.037  \\
GFM             & &                                                                0.77 &                                                                1.03 &                                                                1.03 &                                                               0.113 & &                                                               0.043 &                                                                0.33 &                                                                0.33 &                                                               0.020  \\
Refined GFM     & &                                                                1.20 &                                                                1.02 &                                                                1.02 &                                                               0.178 & &                                                               0.091 &                                                                0.32 &                                                                0.32 &                                                               0.043  \\
 \hline
		\multicolumn{11}{p{0.8\textwidth}}{\scriptsize\emph{Notes:} The first three measures in each panel are normalized by the result for conventional full matching.} \\
	\end{tabular}
}
\end{table}

\begin{sidewaystable}[ht]
	\centering
	\caption{Runtime and memory use by sample size for matching methods.} \label{tab:complexity}
	\vspace{0.15in}
	\resizebox{0.8\textwidth}{!}{%
		\begin{tabular}{l c rrrrr rrrrr rrrrr}
			\multicolumn{17}{l}{\textbf{Panel A: Runtime (in minutes)}} \\
			& &   100 &   500 &    1K &    5K &   10K &   20K &   50K &  100K &  200K &  500K &    1M &    5M &   10M &   50M &  100M \\ \hline
			\\
			Greedy 1:1      & & 0.01 & 0.01 & 0.01 & 0.01 &  0.03 &  0.09 & 0.62 & 2.71 & 12.85 &      &      &      &      &      &       \\
Optimal 1:1     & & 0.06 & 0.06 & 0.08 & 1.08 &  5.03 & 19.40 &      &      &       &      &      &      &      &      &       \\
Replacement 1:1 & & 0.01 & 0.01 & 0.01 & 0.01 &  0.03 &  0.09 & 0.58 & 2.60 & 12.15 &      &      &      &      &      &       \\
Greedy 1:2      & & 0.01 & 0.01 & 0.01 & 0.01 &  0.03 &  0.09 & 0.63 & 2.76 & 13.20 &      &      &      &      &      &       \\
Optimal 1:2     & & 0.06 & 0.06 & 0.10 & 5.02 & 26.19 &       &      &      &       &      &      &      &      &      &       \\
Full matching   & & 0.06 & 0.06 & 0.08 & 0.30 &  0.87 &  2.84 &      &      &       &      &      &      &      &      &       \\
GFM             & & 0.00 & 0.00 & 0.00 & 0.00 &  0.01 &  0.00 & 0.01 & 0.01 &  0.02 & 0.06 & 0.11 & 0.64 & 1.05 & 6.49 & 14.15 \\
Refined GFM     & & 0.00 & 0.00 & 0.00 & 0.01 &  0.01 &  0.01 & 0.01 & 0.02 &  0.03 & 0.08 & 0.16 & 0.96 & 1.52 & 9.31 & 20.07 \\

			\\ \hline
			& \hspace{0.03in} & & & & & & \hspace{0.2in} & & & & \\
			\\
			\multicolumn{17}{l}{\textbf{Panel A: Memory use (in gigabytes)}}  \\
			& &   100 &   500 &    1K &    5K &   10K &   20K &   50K &  100K &  200K &  500K &    1M &    5M &   10M &   50M &  100M \\ \hline
			\\
			Greedy 1:1      & & 0.03 & 0.03 & 0.03 & 0.04 & 0.04 &  0.04 & 0.05 & 0.06 & 0.09 &      &      &      &      &      &       \\
Optimal 1:1     & & 0.12 & 0.13 & 0.15 & 0.74 & 2.75 & 10.21 &      &      &      &      &      &      &      &      &       \\
Replacement 1:1 & & 0.03 & 0.03 & 0.03 & 0.03 & 0.04 &  0.04 & 0.05 & 0.06 & 0.10 &      &      &      &      &      &       \\
Greedy 1:2      & & 0.03 & 0.03 & 0.03 & 0.04 & 0.04 &  0.04 & 0.05 & 0.07 & 0.11 &      &      &      &      &      &       \\
Optimal 1:2     & & 0.12 & 0.13 & 0.15 & 0.74 & 2.74 &       &      &      &      &      &      &      &      &      &       \\
Full matching   & & 0.12 & 0.13 & 0.16 & 0.74 & 2.74 & 10.19 &      &      &      &      &      &      &      &      &       \\
GFM             & & 0.03 & 0.03 & 0.03 & 0.03 & 0.03 &  0.03 & 0.04 & 0.05 & 0.06 & 0.12 & 0.21 & 0.88 & 1.73 & 8.57 & 17.11 \\
Refined GFM     & & 0.03 & 0.03 & 0.03 & 0.03 & 0.03 &  0.03 & 0.04 & 0.05 & 0.06 & 0.12 & 0.21 & 0.88 & 1.73 & 8.57 & 17.11 \\

			\\ \hline
			\multicolumn{17}{p{\textwidth}}{\footnotesize\emph{Notes:} Each cell presents the runtime and memory of the matching implementations for different sample sizes. The first panel shows runtime in minutes, and the second panel shows memory use in gigabytes. Each column represents a different sample size where ``K'' denotes thousand and ``M'' denotes million. The rows indicate matching method. Greedy 1:1, Replacement 1:1 and Greedy 1:2 are implemented by the \texttt{Matching} package. Optimal 1:1, Optimal 1:2 and conventional full matching are implemented by the \texttt{optmatch} package. Generalized full matching (GFM) is implemented by the \texttt{scclust} package. Blank cells indicate that the corresponding matching method did not terminate successfully for the corresponding sample size within reasonable time and memory limits. Each measure is based on 1,000 simulation rounds.} \\
		\end{tabular}
	}
\end{sidewaystable}

\clearpage

\section{Additional extrapolation results}

The following tables provide unadjusted and adjusted covariate averages for all treatment conditions in the full population and in the subpopulation of voters in the 2004 partisan election.

\begin{table}[ht]
	\centering
	\caption{Covariate balance before and after matching adjustment in full QVF population.}
	\resizebox{0.9\textwidth}{!}{%
	\begin{tabular}{l c rrrrrr}
		\multicolumn{8}{l}{\textbf{Panel A: Covariate balance before matching}} \\
		& & \multicolumn{1}{r}{Control} & \multicolumn{1}{r}{Civic Duty} & \multicolumn{1}{r}{Hawthorne} & \multicolumn{1}{r}{Self} & \multicolumn{1}{r}{Neighbors} & \multicolumn{1}{r}{Non-experiment} \\ \hline
		Birth year &  & 1956.19 & 1956.34 & 1956.30 & 1956.21 & 1956.15 & 1957.96 \\
Female (\%) &  &   49.89 &   50.02 &   49.90 &   49.96 &   50.00 &   53.32 \\
Voted Aug 2000 (\%) &  &   25.19 &   25.36 &   25.04 &   25.11 &   25.12 &   14.65 \\
Voted Aug 2002 (\%) &  &   38.94 &   38.88 &   39.43 &   39.19 &   38.66 &   22.59 \\
Voted Aug 2004 (\%) &  &   40.03 &   39.94 &   40.32 &   40.25 &   40.67 &   18.71 \\
Voted Nov 2000 (\%) &  &   84.34 &   84.17 &   84.44 &   84.04 &   84.17 &   52.49 \\
Voted Nov 2002 (\%) &  &   81.09 &   81.11 &   81.30 &   81.15 &   81.13 &   41.93 \\
Voted Nov 2004 (\%) &  &  100.00 &  100.00 &  100.00 &  100.00 &  100.00 &   67.57 \\
 \hline
		\\
		\\
		\multicolumn{8}{l}{\textbf{Panel B: Covariate balance after matching}} \\
		& & \multicolumn{1}{r}{Control} & \multicolumn{1}{r}{Civic Duty} & \multicolumn{1}{r}{Hawthorne} & \multicolumn{1}{r}{Self} & \multicolumn{1}{r}{Neighbors} & \multicolumn{1}{r}{Non-experiment} \\ \hline
		Birth year &  & 1958.16 & 1958.49 & 1958.44 & 1958.57 & 1958.51 & 1957.87 \\
Female (\%) &  &   53.29 &   53.28 &   53.28 &   53.29 &   53.28 &   53.15 \\
Voted Aug 2000 (\%) &  &   15.19 &   15.19 &   15.19 &   15.19 &   15.19 &   15.19 \\
Voted Aug 2002 (\%) &  &   23.42 &   23.42 &   23.42 &   23.42 &   23.42 &   23.43 \\
Voted Aug 2004 (\%) &  &   19.80 &   19.80 &   19.80 &   19.80 &   19.80 &   19.80 \\
Voted Nov 2000 (\%) &  &   54.11 &   54.14 &   54.13 &   54.13 &   54.13 &   54.11 \\
Voted Nov 2002 (\%) &  &   43.94 &   43.94 &   43.94 &   43.94 &   43.94 &   43.92 \\
Voted Nov 2004 (\%) &  &  100.00 &  100.00 &  100.00 &  100.00 &  100.00 &   68.76 \\
 \hline
	\end{tabular}
	}
\end{table}

\begin{table}[ht]
	\centering
	\caption{Covariate balance before and after adjustment among voters in 2004 partisan election.}
	\resizebox{0.9\textwidth}{!}{%
	\begin{tabular}{l c rrrrrr}
		\multicolumn{8}{l}{\textbf{Panel A: Covariate balance before matching}} \\
		& & \multicolumn{1}{r}{Control} & \multicolumn{1}{r}{Civic Duty} & \multicolumn{1}{r}{Hawthorne} & \multicolumn{1}{r}{Self} & \multicolumn{1}{r}{Neighbors} & \multicolumn{1}{r}{Non-experiment} \\ \hline
		Birth year &  & 1956.19 & 1956.34 & 1956.30 & 1956.21 & 1956.15 & 1955.71 \\
Female (\%) &  &   49.89 &   50.02 &   49.90 &   49.96 &   50.00 &   54.51 \\
Voted Aug 2000 (\%) &  &   25.19 &   25.36 &   25.04 &   25.11 &   25.12 &   20.49 \\
Voted Aug 2002 (\%) &  &   38.94 &   38.88 &   39.43 &   39.19 &   38.66 &   31.94 \\
Voted Aug 2004 (\%) &  &   40.03 &   39.94 &   40.32 &   40.25 &   40.67 &   26.94 \\
Voted Nov 2000 (\%) &  &   84.34 &   84.17 &   84.44 &   84.04 &   84.17 &   70.16 \\
Voted Nov 2002 (\%) &  &   81.09 &   81.11 &   81.30 &   81.15 &   81.13 &   58.81 \\
Voted Nov 2004 (\%) &  &  100.00 &  100.00 &  100.00 &  100.00 &  100.00 &  100.00 \\
 \hline
		\\
		\\
		\multicolumn{8}{l}{\textbf{Panel B: Covariate balance after matching}} \\
		& & \multicolumn{1}{r}{Control} & \multicolumn{1}{r}{Civic Duty} & \multicolumn{1}{r}{Hawthorne} & \multicolumn{1}{r}{Self} & \multicolumn{1}{r}{Neighbors} & \multicolumn{1}{r}{Non-experiment} \\ \hline
		Birth year &  & 1955.90 & 1956.11 & 1956.25 & 1956.26 & 1956.20 & 1955.74 \\
Female (\%) &  &   54.17 &   54.16 &   54.17 &   54.17 &   54.16 &   54.17 \\
Voted Aug 2000 (\%) &  &   20.83 &   20.83 &   20.81 &   20.83 &   20.83 &   20.83 \\
Voted Aug 2002 (\%) &  &   32.46 &   32.46 &   32.46 &   32.46 &   32.46 &   32.46 \\
Voted Aug 2004 (\%) &  &   27.92 &   27.92 &   27.92 &   27.92 &   27.92 &   27.92 \\
Voted Nov 2000 (\%) &  &   71.20 &   71.22 &   71.21 &   71.20 &   71.21 &   71.20 \\
Voted Nov 2002 (\%) &  &   60.44 &   60.45 &   60.45 &   60.45 &   60.45 &   60.45 \\
Voted Nov 2004 (\%) &  &  100.00 &  100.00 &  100.00 &  100.00 &  100.00 &  100.00 \\
 \hline
	\end{tabular}
	}
\end{table}

\clearpage

\end{document}